\documentclass[12pt]{article}

\usepackage{fancyhdr,soul}
\usepackage[retainorgcmds]{IEEEtrantools}
\usepackage{isomath}
\usepackage{amsmath,amsthm}
\usepackage{amsbsy}
\usepackage{amssymb}
\usepackage{amscd}
\usepackage{amsfonts}
\usepackage{graphics}
\usepackage{verbatim}
\usepackage{subfigure}
\usepackage{xspace}
\usepackage{euscript}
\usepackage{alltt}
\usepackage{boxedminipage}
\usepackage{float}
\usepackage{color}
\usepackage[,colorlinks,urlcolor=blue,citecolor=blue]{hyperref}
\usepackage[all]{xy}
\usepackage{t1enc}
\usepackage{times,exscale}
\usepackage{graphicx,calc}
\usepackage{mdwlist}
\usepackage{stmaryrd}
\usepackage{units}
\usepackage{setspace}
\usepackage{isomath}
\usepackage{amsmath}
\usepackage{amssymb}
\usepackage{amscd}
\usepackage{amsfonts}
\usepackage{tikz}

\newcommand{\R}{\mathbb R}

\newcommand{\Z}{\mathbb{Z}}

\newcommand{\bfe}{{\mathbold e}}
\newcommand{\bff}{{\mathbold f}}
\newcommand{\bfg}{{\mathbold g}}
\newcommand{\bfh}{{\mathbold h}}

\newcommand{\bfn}{{\mathbold n}}

\newcommand{\bfp}{{\mathbold p}}
\newcommand{\bfq}{{\mathbold q}}
\newcommand{\bfr}{{\mathbold r}}

\newcommand{\bfu}{{\mathbold u}}
\newcommand{\bfv}{{\mathbold v}}

\newcommand{\bfx}{{\mathbold x}}
\newcommand{\bfy}{{\mathbold y}}
\newcommand{\bfz}{{\mathbold z}}

\newcommand{\bfA}{{\mathbold A}}
\newcommand{\bfB}{{\mathbold B}}

\newcommand{\bfD}{{\mathbold D}}
\newcommand{\bfE}{{\mathbold E}}
\newcommand{\bfF}{{\mathbold F}}
\newcommand{\bfG}{{\mathbold G}}

\newcommand{\bfI}{{\mathbold I}}

\newcommand{\bfP}{{\mathbold P}}

\newcommand{\bfS}{{\mathbold S}}

\newcommand{\bfX}{{\mathbold X}}

\newcommand{\beq}{\begin{equation}}
\newcommand{\eeq}{\end{equation}}
\newcommand{\beqs}{\begin{eqnarray}}
\newcommand{\eeqs}{\end{eqnarray}}
\newcommand{\beql}{\begin{equation} \label}
\newcommand{\half}{\frac{1}{2}}

\newtheorem{proposition}{Proposition}[section]

\newcommand{\bfzeta}{{\boldsymbol{\zeta}}}
\newcommand{\bftheta}{{\boldsymbol{\theta}}}

\newcommand{\bfzero}{\mathbf{0}}

\newcommand{\grad}{\mathop{\rm grad}\nolimits}
\newcommand{\divergence}{\mathop{\rm div}\nolimits}
\newcommand{\curl}{\mathop{\rm curl}\nolimits}

\DeclareMathOperator{\tr}{tr}

\newcommand{\bbS}{{\mathbb{S}}}

\newcommand*\circled[1]{\tikz[baseline=(char.base)]{
            \node[shape=circle,draw,inner sep=2pt] (char) {#1};}}

\usepackage[margin=20mm]{geometry}
\numberwithin{equation}{section}
\title{Atomistic-to-Continuum Multiscale Modeling with Long-Range Electrostatic Interactions in Ionic Solids}
\author{Jason Marshall\footnote{jmarshal@andrew.cmu.edu} \  and Kaushik Dayal\footnote{kaushik@cmu.edu} \\ {\small Carnegie Mellon University}}
\date{\today}

\begin{document}
\pagestyle{fancyplain}
\lhead{\fancyplain{\scriptsize Atomistic-to-Continuum Multiscale Modeling with Long-Range Electrostatic Interactions in Ionic Solids }
{\scriptsize Atomistic-to-Continuum Multiscale Modeling with Long-Range Electrostatic Interactions in Ionic Solids}}
\rhead{\fancyplain{\scriptsize Jason Marshall, Kaushik Dayal}{\scriptsize Jason Marshall, Kaushik Dayal}}
\maketitle

\begin{large}
\begin{center}
\textbf{\mbox{Rodney Hill Anniversary Issue of the Journal of the Mechanics and Physics of Solids}}
\end{center}
\end{large}
%%%%%%%%%%%%%%%%%%%%%
%%%%%%%%%%%%%%%%%%%%%
%%%%%%%%%%%%%%%%%%%%%
%%%%%%%%%%%%%%%%%%%%%

\begin{abstract}
We present a multiscale atomistic-to-continuum method for ionic crystals with defects.
Defects often play a central role in ionic and electronic solids, not only to limit reliability, but more importantly to enable the functionalities that make these materials of critical importance.
Examples include solid electrolytes that conduct current through the motion of charged point defects, and complex oxide ferroelectrics that display multifunctionality through the motion of domain wall defects.
Therefore, it is important to understand the structure of defects and their response to electrical and mechanical fields.
A central hurdle, however, is that interactions in ionic solids include both short-range atomic interactions as well as long-range electrostatic interactions.
Existing atomistic-to-continuum multiscale methods, such as the Quasicontinuum method, are applicable only when the atomic interactions are short-range.
In addition, empirical reductions of quantum mechanics to density functional models are unable to capture key phenomena of interest in these materials.

To address this open problem, we develop a multiscale atomistic method to coarse-grain the long-range electrical interactions in ionic crystals with defects.
In these settings, the charge density is rapidly varying, but in an almost-periodic manner.
The key idea is to use the polarization density field as a multiscale mediator that enables efficient coarse-graining by exploiting the almost-periodic nature of the variation.
In regions far from the defect, where the crystal is close-to-perfect, the polarization field serves as a proxy that enables us to avoid accounting for the details of the charge variation.
We combine this approach for long-range electrostatics with the standard Quasicontinuum method for short-range interactions to achieve an efficient multiscale atomistic-to-continuum method.
As a side note, we examine an important issue that is critical to our method: namely, the dependence of the computed polarization field on the choice of unit cell.
Potentially, this is fatal to our coarse-graining scheme; however, we show that consistently accounting for boundary charges leaves the continuum electrostatic fields invariant to choice of unit cell.


\end{abstract}

\tableofcontents

%%%%%%%%%%%%%%%%%%%%%
%%%%%%%%%%%%%%%%%%%%%
%%%%%%%%%%%%%%%%%%%%%
%%%%%%%%%%%%%%%%%%%%%

\section{Introduction}

Ionic crystals such as solid electrolytes and complex oxides are central to modern technologies for energy storage, sensing, actuation, and other functional applications.
Atomic-scale defects often play a central role in these materials, not only to limit reliability, but more importantly to enable the functionalities that make these materials of critical importance.
E.g., in solid electrolytes, conduction is mediated by charged point defects \cite{solid-electrolytes}; and in complex oxide ferroelectrics, functionality is mediated by planar domain wall defects, and loss of functionality often occurs when domain walls are ``pinned'' by charged oxygen vacancy point defects \cite{ferroelectrics-scott-book}.
A fundamental understanding of these materials therefore requires an accounting of the atomic-level structure of the defects.
This poses a multiscale problem: atomic-level resolution is required at the defect, while complex geometries and boundary conditions require the modeling of a large specimen.

While defects play a critical role in determining properties, they occupy a tiny volume of the lattice; an exceedingly large fraction of the crystal is close-to-perfect.
This feature is exploited in most leading atomistic multiscale methods such as the Quasicontinuum (QC) method \cite{QC-review1,QC-review2}: typically, an adaptive coarse-graining is used with atomic resolution in the vicinity of the defect and a coarse-grained description further away as the crystal tends to a perfect lattice.
Another important aspect of the coarse-graining is the use of sampling or quadrature to efficiently evaluate the energy in the coarse-grained region.
It is essential that the atomic interactions are short-range to allow the evaluation of the energy at the quadrature points to be efficient.
Therefore, existing multiscale methods cannot handle long-range electrostatic/ionic interactions that decay as $1/r$, where $r$ is the separation between charges.

A symptom of this difficulty with electrostatic interactions can be observed in standard proofs of the Cauchy-Born (CB) theorem that require the interactions to decay faster than $1/r^3$ \cite{friesecke-james, blanc-lebris-lions}.
Roughly, this implies that the standard CB theorem requires that the charge distribution in each unit cell of the lattice should not have net charge or net dipole character, but only higher-order multipoles.
As shown in \cite{james-muller}, it is not possible to define a meaningful energy density $W(\cdot)$ in such a setting, i.e., the standard decomposition of the energy used in elasticity
\[
 	E = \int_\Omega W(\cdot) \ d\Omega - \int_{\partial\Omega} \text{boundary working}
\]
is not valid.
When long-range electrostatic forces are involved, $W(\cdot)$ does not depend solely on the local value of a field (the strain field in elasticity or the polarization field in electrostatics).
Rather, it depends on the electrostatic fields in a nonlocal manner as well as boundary conditions.
While not always stated explicitly, some notion of the CB theorem is inherent in most atomistic multiscale formulations.

The restrictions described above on the nature of atomic-level interactions for conventional multiscale methods exclude an extremely large class of materials; essentially, all dielectrics, polarizable solids, and ionic solids.
In dielectrics and polarizable solids, the non-vanishing dipole moment in the unit cell is central to the physics of dielectric response, spontaneous polarization.
In ionic solids such as ionic conductors, the existence of charged defects is central to enabling conduction.
Therefore, it is essential to develop methods that can handle long-range electrostatic interactions.

In this paper, we present a multiscale method that is tailored to allow both short-range atomic interactions as well as long-range electrostatic interactions.
Some key features of these interactions are as follows.
The short-range atomic interactions can be highly nonlinear and involve complex multibody interactions; however, they are typically restricted to 2nd or 3rd nearest-neighbors in a lattice.
The long-range electrostatic interactions have the opposite features: the interactions between charges are entirely pairwise, but the interactions between {\em every} pair of charges in the system can be non-negligible.
A further important feature is that the short-range and electrostatic contributions to the total energy combine additively.
These features enable us to leverage much of the existing work for short-range interactions.
In particular, we can use a standard version of the Quasicontinuum method \cite{Tadmor} for the short-range interactions in combination with a method that we develop for efficiently computing the electrostatic interactions.

While the presentation of our method in this paper is largely formal, key aspects build on -- and are supported by -- rigorous results of James and M\"{u}ller  \cite{james-muller} and others following their work, e.g. \cite{schlomerkemper-schmidt}.
We note that seminal formal results in this topic were earlier obtained by \cite{toupin-elastic-dielectric}.
While these works deal with point dipoles arranged in a lattice, our focus is on charges; however, due to the fact that the net charge in each unit cell of the lattice is $0$, many of the key results carry over largely unchanged, as also noticed previously by \cite{Xiao_thesis, puri-bhatta}.
The central idea that we exploit is that a charged lattice can be coarse-grained by introducing a polarization density field.
I.e., electrostatic quantities that in principle require the solution of a Poisson problem with a rapidly and almost-periodically oscillating forcing term due to charge density can instead be computed with a much smoother forcing that is related to the polarization density field.

Efficient and accurate methods for interactions in charged systems have a long history in numerical methods.
The key challenge is the long-range nature of interactions: in principle, a system of $N$ point charges requires $O(N^2)$ calculations.
For $N$ of the order required for typical problems of interest today, this is completely infeasible.
However, the seminal and beautiful Fast Multipole Method (FMM) of Greengard and Rokhlin \cite{greengard1987} provided a breakthrough in enabling this in $O(N)$ calculations with a controlled error.
A key strength of the FMM is that the charge distribution can be completely arbitrary and non-uniform; however, this generality also means that the method does not exploit the structure in a given problem.
As noted above, in the problems of relevance here, the crystalline structure is lost in the vicinity of the defect, but large parts of the crystal are almost perfect.
The FMM, however, is unable to exploit this structure, whereas our method is more tailored (and thereby also less general) and appears to scale almost independent of $N$ asymptotically.
Another leading method for atomic-level calculations with electrostatic interactions is the Ewald method (described in e.g. \cite{Tuckerman-MD-book}).
It is restricted to periodic settings and therefore inapplicable to multiscale calculations with defects and complex geometries and boundary conditions.

As mentioned above, our coarse-graining of electrostatic interactions is based on the notion of a polarization density field.
However, it is well-known, e.g. \cite{Resta-Vanderbilt}, that the polarization density of a periodic solid depends on the choice of unit cell.
At first sight, this is a disturbing observation and much work has been done in the materials physics community on using quantum mechanical notions such as the Berry phase to obtain a unique choice of polarization for a given periodic charge distribution \cite{Resta-Vanderbilt}.
However, an important aspect of that approach is the insistence on starting from an infinite periodic solid.
In both the formal calculations presented here, and the related rigorous calculations in the references above, the starting point is a finite periodic solid whose limit behavior is studied.
From this ``real-space'' perspective, the boundaries of the crystal lattice enter naturally into the problem, in sharp contrast to starting from the infinite periodic solid where boundaries are ill-defined.
The critical importance of the boundaries is that they, roughly speaking, compensate for the choice of unit cell.
I.e., while different unit cell choices lead to different expressions for the polarization density, these also lead to different bound surface charges on the boundaries.
When {\em both} the bulk bound charge and the surface bound charge are consistently accounted for in the calculations, the electric field and other quantities of relevance to the energy do not depend on the choice of unit cell up to an error that scales with size of the lattice and vanishes in the limit.
Therefore, we take the view that a unique choice of unit cell to compute polarization density is unnecessary and any choice of unit cell is -- in principle -- equally valid.
I.e., the polarization is an intermediate coarse-grained quantity, but there is no fundamental physical reason to have a specific choice.
In practice, notions of crystal symmetry typically are most useful in selecting a unit cell.
Heuristically, this perspective is comparable to the universally-accepted view in continuum mechanics that any reference configuration is -- in principle -- equally valid.
While certain choices of reference configuration can lead to conceptual and algebraic simplifications, there is no fundamental physical preference for any specific choice.
What is more relevant is that it is possible to go between different choices with appropriate transformations to the kinematic variables and the energy densities.

The paper is organized as follows.
\begin{itemize}

	\item In Section \ref{sec:formulation}, we formulate the problem at the atomic level and briefly describe the well-developed QC approach to handle short-range interactions.

	\item In Section \ref{sec:electrostatics}, we describe our treatment of the long-range interactions.  Formally, we show the appearance of the polarization as an intermediary multiscale quantity to link atomistic charge distributions with coarse-grained fields.  We also examine the issue of the choice of unit cell for polarization;  accounting consistently for the boundaries for a given unit cell does not affect the coarse-graining.

	\item In Section \ref{sec:implementation}, we outline the kinematic coarse-graining that follows the complex local QC method \cite{Tadmor} and other aspects of the numerical implementation.

	\item In Section \ref{sec:examples}, we outline the model material and its response to a variety of electrical and mechanical loadings.

	\item In Section \ref{sec:conclusion}, we discuss various aspects of the work including open problems for ongoing and future work.

\end{itemize}

%%%%%%%%%%%%%%%%%%%%%
%%%%%%%%%%%%%%%%%%%%%
%%%%%%%%%%%%%%%%%%%%%
%%%%%%%%%%%%%%%%%%%%%

\subsection{Notation}
\label{sec:notation}

Throughout the paper, bold lowercase and uppercase letters denote vectors and tensors.
The summation convention is {\em not} used in this paper.
Sums will be explicitly written out to avoid confusion except where stated.

We define $L$ as a Bravais lattice with three independent lattice vectors that make up a unit cell.
\begin{equation}
	\label{eq:L_def}
	L(\bfe_i,\bfzero) = \left( \bfx \in \R^3 \text{, }\bfx=\sum_i \nu^i\bfe_i \text{ where }\nu^i \in \Z \text{, }i=1,2,3\right)
\end{equation} 

\begin{tabular}{ l c l }
$E_{total}$  & = & Total energy\\
$E$ & = & Electrostatic energy\\
$U$ & = & Interatomic potential energy (i.e. Lennard-Jones, Buckingham types)\\
$W$ & = & Short-range strain energy density (related to $U$)\\
$\Omega$ & = & Continuum body in the current configuration \\
$\Omega_0$ & = & Continuum body in the reference configuration \\
$Q^s$ & = & Charge of the atomic species indexed by $s$\\
$\rho$ & = & Charge density field in $\Omega$\\
$\epsilon_0$ & = & dielectric constant for vacuum \\
$\mathbb{K}$ & = & Dipole-dipole interaction electrostatic kernel\\

$\bfu$ & = & Displacement field\\
$\bfx$ & = & Slow variable representing position in current configuration\\
$\bfy$ & = & Fast variable representing position in current configuration\\
$\bfx_0$ & = & Position in reference configuration\\

$\bfzeta^s$ & = & Intra-unit cell position of species $s$, defined in the reference \\
$\bfp$ & = & Polarization density field\\
$\phi$ & = & Electric potential\\
$\grad_{\bfx} \phi$ & = & Electric field\\
$\bfF$ & = & Deformation gradient\\
$J$ & = & $\det \bfF$ \\
$\square_i$& = & $i$th unit cell\\
$\triangle_i$ & = & $i$th partial unit cell\\
$\epsilon$ & = & Continuum material point lengthscale\\
$l$ & = & Atomic lengthscale\\
$L$ & = & Lengthscale over which continuum fields vary\\
$\mathcal{B}_{\epsilon}$ & = & Ball of radius $\epsilon$\\
$\mathcal{D}_{\epsilon}$ & = & 2D disk of radius $\epsilon$\\
$\grad_{\bfx_0}$ & = & Gradient with respect to the reference configuration\\
$\grad_{\bfx}$ & = & Gradient with respect to the current configuration\\
$\sigma$ & = & Surface charge due to non-neutral partial unit cells.
\end{tabular}

$\Omega_{\#}$ and $\Omega_{\square}$ represent the decomposition of $\Omega$ into the partial unit cells on the boundary ($\Omega_{\#}$), and the remainder $\Omega_{\square} := \Omega \backslash \Omega_{\#}$; see Fig. \ref{fig:omega-decomposition}.

%%%%%%%%%%%%%%%%%%%%%
%%%%%%%%%%%%%%%%%%%%%
%%%%%%%%%%%%%%%%%%%%%
%%%%%%%%%%%%%%%%%%%%%

\section{Problem Formulation}
\label{sec:formulation}

We consider a crystal occupying a region $\Omega$, composed of charged species indexed by $s$, each carrying a fixed charge $Q^s$.
The notation of species is used broadly; it refers to ions and electrons, as well as electron ``shells'' as used in core-shell models \cite{core-shell-ref}.
We assume that the charges are all point charges, e.g. nuclei, or that they can be represented through a center of charge as in electron shells.
Therefore, the charge distribution $\rho(\bfx)$ is a collection of Dirac masses.

The total energy in the body can be written in the form below:
\begin{equation}
\label{eq:energy_lattice}
	E_{total} =
		\underbrace{ \sum_{\substack{i \\ i \neq j}} U_i \left(\{\bfr_{ij} \} \right)}_\text{short-range}
		+
		\underbrace{\half \sum_{\substack{i,j \\ i \neq j}} \frac{Q^s_i Q^s_j}{4 \pi \epsilon_0\vert \bfr_{ij}\vert}}_\text{long-range}
\end{equation}
$\bfr_{ij}$ is the vector between charges $i$ and $j$, and $Q^s_i$ is the charge carried by $i$th atom of species $s$.
The function $U$ is the given short-range interatomic potential and can typically involve multibody interactions.

We restrict our attention to zero temperature.
Our goal is to find local minimizers of $E_{total}$ to obtain the equilibrium structure subject to applied mechanical and electrostatic loadings.
Brute force minimization is infeasible even for the short-range contributions for realistically large systems.
This motivated the QC method and related approaches \cite{QC-review1,QC-review2} for short-range interactions.
We will use the so-called local QC multi-lattice method for the short-range energy largely following \cite{Tadmor}.
As noted above, there are two ingredients to this multiscale approach: first, a kinematic condensation of the degrees of freedom using interpolations, and second, efficient calculation of the energy sum by using sampling or quadrature in relatively uniform regions.
The term {\em local} refers to the fact that we use a sampling approximation {\em everywhere} in the specimen including at the defect core where it is likely to be quite inaccurate.

We begin with the species in the reference configuration arranged in a periodic multi-lattice.
The unit cell is denoted $\square$ and the atomic length scale $l$ (Figure \ref{fig:unit_cell_omega}).
In a perfect lattice with short-range interactions, the energy converges to $\int_{\Omega} W(\grad_{\bfx_0} \bfu,\bfzeta^s) \ dV_{\bfx_0}$.
Here, $W$ is the strain energy density, and  $\grad_{\bfx_0} \bfu$ and $\bfzeta^s$ are the deformation gradient and the ``shifts'' or relative displacements between lattices \cite{blanc-lebris-lions,friesecke-james}.
While the expression for $W$ is algebraically involved, it is conceptually simple and comes directly from $U$.
In a perfect multi-lattice, there is a well-defined notion of energy per atom since every atom of a given species is in the same environment.
Therefore, it is possible to define an energy density by finding the energy of the atoms in a unit cell and dividing the cell volume, which is precisely $W$.
The energy naturally depends on the shape of the unit cell and the positions of the different species within it, and this information is contained in $\grad_{\bfx_0} \bfu$ and $\bfzeta^s$ respectively.
The QC method replaces the sum in (\ref{eq:energy_lattice}) by sampling the energy density $W(\grad_{\bfx_0} \bfu,\bfzeta^s)$ and using appropriate weights.
In the more sophisticated formulations of QC, the energy is computed without this approximation in highly-distorted regions such as the vicinity of the defect \cite{knap-ortiz}.
In the local QC, the approximation is used throughout the specimen, including at the defect.
For further details, we refer the reader to recent reviews of the extensive literature on applying QC to materials with short-range interactions \cite{QC-review1,QC-review2}.

\begin{figure}[ht!]
	\centering
	\includegraphics[width=125mm]{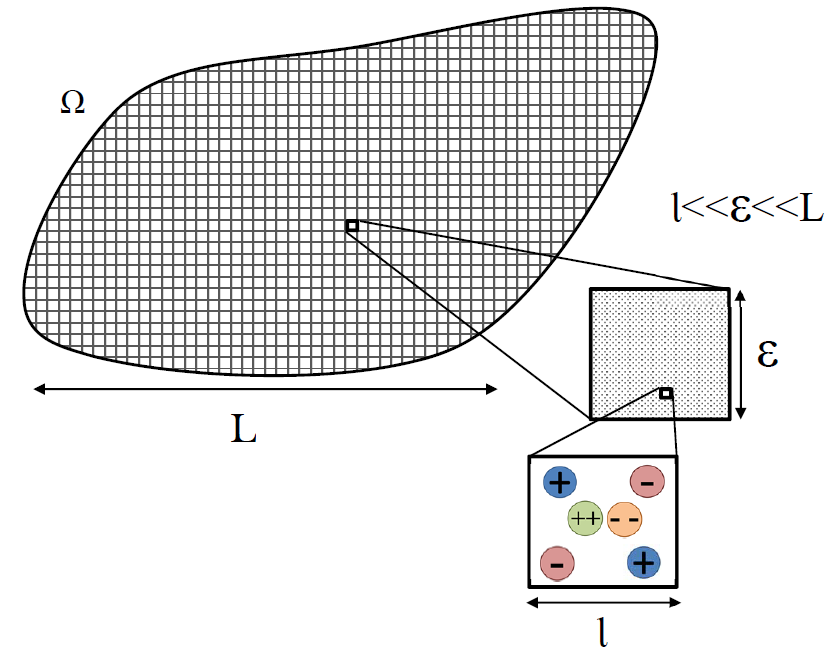}
	\caption{\small Domain $\Omega$ showing the separation of scales with sample charge distribution.}
	\label{fig:unit_cell_omega}
\end{figure}

%%%%%%%%%%%%%%%%%%%%%
%%%%%%%%%%%%%%%%%%%%%
%%%%%%%%%%%%%%%%%%%%%
%%%%%%%%%%%%%%%%%%%%%

\section{Electrostatic Interactions}
\label{sec:electrostatics}

In this section, we describe the formal calculations that enable us to efficiently account for the long-range electrostatic interactions.
A key aspect is the appearance of the polarization density as a multiscale mediator.
Because our treatment here is formal, we go between charge density fields and point charges as convenient by assuming that our calculations are valid even when the charge density field consists of Dirac masses.
Rigorous treatments of many of the key aspects are available in the literature \cite{james-muller} and also are the focus of our ongoing work.

%%%%%%%%%%%%%%%%%%%%%
%%%%%%%%%%%%%%%%%%%%%
%%%%%%%%%%%%%%%%%%%%%
%%%%%%%%%%%%%%%%%%%%%

\subsection{Why the Electrostatic Energy is Long-Range}
\label{toy-example}

We construct and examine some simple examples to understand why the electrostatics is denoted ``long-range''.
E.g., the Lennard-Jones potential has interactions that decay as $r^{-6}$ and these interactions nominally extend to $\infty$.
As we see, there are important differences when the interactions decay as $r^{-1}$ in electrostatics.

Consider a uniform lattice and 3 cases of charge arrangement within in each unit cell: (i) a charge, (ii) a pair of equal and opposite charges forming a dipole, and (iii) two pairs of equal and opposite charges that form a quadrupole with zero dipole moment (Fig. \ref{fig:toy-model}).

\begin{figure}[ht!]
	\centering
	\includegraphics[width=170mm]{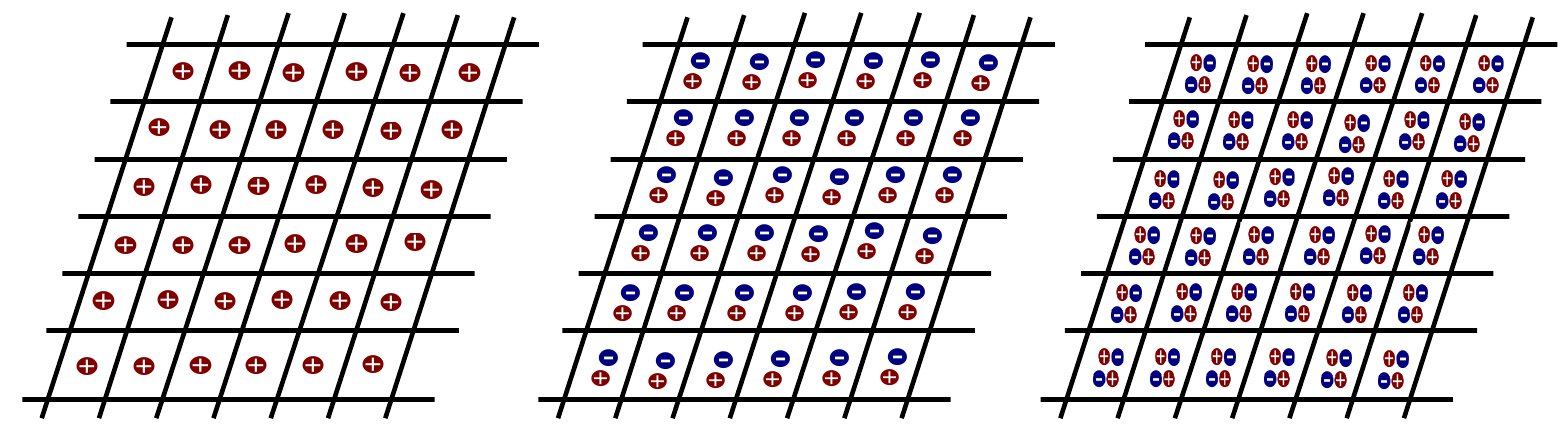}
	\caption{\small 3 cases of charge arrangement: (i) net charge, (ii) net dipole but no net charge, and (iii) net quadrupole, but no net charge and no net dipole.}
	\label{fig:toy-model}
\end{figure}

We first consider the lattice of charges.
As a rough measure of energy density, we compute the energy of the charge in the chosen unit cell due to its interaction with all the other unit cells in the body.
This is the product between the magnitude of the charge in the chosen unit cell with the electrostatic potential created by the rest of the body.
The potential due to a charge at a distance $r$ from the chosen unit cell scales as $r^{-1}$.
Further, at a distance $r$ from the chosen unit cell, we consider a spatial region with the shape of a spherical shell with unit thickness.
This shell has volume that scales as $r^2$ and therefore roughly contains $r^2$ charges.
So the total potential at the chosen unit cell due to the rest of the system is $\sum_{r=1}^{\infty} \frac{1}{r} r^2 = \sum_{r=1}^{\infty} r \rightarrow \infty$.
That is, the energy density of the body is unbounded in the large-body limit.
The physical implication of this calculation is that large clusters of unbalanced charges have extremely high-energy and are thus unlikely to be observed in real materials.

Next consider the lattice of dipoles.
As a rough measure of energy density, we compute the energy of the dipole in the chosen unit cell due to its interaction with all the other unit cells in the body.
This is the product between the magnitude of the dipole in the chosen unit cell with the electrostatic field created by the rest of the body.
The electric field due to a dipole at a distance $r$ from the chosen unit cell scales as $r^{-3}$.
Further, the shell at a distance $r$ again contains roughly $r^2$ dipoles.
So the total electric field at the chosen unit cell due to the rest of the body is $\sum_{r=1}^{\infty} \frac{1}{r^3} r^2 = \sum_{r=1}^{\infty} \frac{1}{r}$.
This sum nominally also tends to infinity.
However, the issue is more subtle.
The full expression for the electric field due to a unit dipole oriented in the direction $\hat{\bfn}$ at a position $\bfx$ is $E_i = \sum_j\frac{\delta_{ij} - 3 \hat{\bfx}_i \hat{\bfx}_j}{4\pi |\bfx|^3} \hat{n}_j$.
Certain components of the summation have alternating sign, and this leads to conditionally convergent sums, and in general this sum is at the border of convergence / divergence.
The physical implication of this calculation is that the energy density of a large collection of dipoles is extremely sensitive to the precise boundary conditions that are imposed far away ``at infinity''.
Alternatively, in a finite body, this says that the energy {\em density} at a given point is extremely sensitive to the distribution of dipoles throughout the entire body.
This physical implication is the reason to denote electrostatic interactions as ``long-range'', namely, it is not possible to define a meaningful energy density that depends only on the local state of the crystal.
The energy density at a given material point instead requires accounting for the state of the body at every other material point as well as boundary conditions.
As we see below, this also poses practical difficulties for standard multiscale algorithms such as QC.
This makes these methods inapplicable to an extremely broad class of solids: in all dielectrics, polarizable media, and ionic solids, the non-vanishing dipole moment in the unit cell is central to the physics of dielectric response, spontaneous polarization, and other key electrical properties.

Finally, consider the lattice of quadrupoles.
As a rough measure of energy density, we compute the energy of the quadrupole in the chosen unit cell due to its interaction with all the other unit cells in the body.
This is the product between the magnitude of the quadrupole in the chosen unit cell with the gradient of the electrostatic field created by the rest of the body.
The electric field due to a quadrupole at a distance $r$ from the chosen unit cell scales as $r^{-5}$.
Further, the shell at a distance $r$ contains again roughly contains $r^2$ quadrupoles.
So the total electric field at the chosen unit cell due to the rest of the body is $\sum_{r=1}^{\infty} \frac{1}{r^5} r^2 = \sum_{r=1}^{\infty} \frac{1}{r^3}$.
This sum converges rapidly.
This setting corresponds to the case of metals and other systems with mobile electrons that allow the charge to redistribute itself to ``shield'' the dipole moment.
This leads effectively to short-range interactions: though nominally the interactions are present at all values of $r$, the rapid convergence of the series allows truncation at finite cut-off without significant error.
For this reason, among others, short-range potentials with interactions involving only nearest- and next-nearest-neighbors are sufficiently accurate to model metallic and related systems.
%%%%%%%%%%%%%%%%%%%%%
%%%%%%%%%%%%%%%%%%%%%
%%%%%%%%%%%%%%%%%%%%%
%%%%%%%%%%%%%%%%%%%%%

\subsection{Existing Numerical Approaches to Compute Electrostatic Interactions}

The long-range nature of the electrical interactions described above leads to practical hurdles in atomic multiscale computations.
Leading methods to handle these interactions are Ewald sums and the Fast Multipole Method (FMM).
The Ewald method \cite{griebel-MD-book} assumes perfect periodicity.
This is appropriate only for perfect crystals.
Approximating defect calculations by periodic supercells has severe artifacts even with purely short-range interactions, a difficulty much more pronounced when interactions are long-range in nature.
The fast multipole method (FMM) reduces the problem from $O(N^2)$ to $O(N)$, but is still extremely expensive with atomic multiscale calculations in crystals often as large as $N\sim 10^{21}$.
In addition, the ability of FMM to deal with arbitrary charge distributions also implies that it does not exploit the close-to-uniform distortion away from the defect \cite{beatson-greengard-FMM-review}.

For short-range interactions, multiscale atomistic methods such as the QC method borrow ideas of quadrature rules from FEM \cite{QC-review1,QC-review2} to evaluate the energy at various sampling/quadrature atoms and then use quadrature weights.
This idea depends critically on the energy evaluation at the quadrature point being a fast calculation.
This is {\em not} a fast calculation if the quadrature charges interact directly with every other charge in the system.
Therefore, these multiscale methods are applicable only to materials with short-range interactions.
Multiscale QC-based methods for Orbital-Free Density Functional Theory -- an empirical simplification of Density Functional Theory for metallic systems -- use a continuous charge density rather than discrete point charges, but formally the issues are the same.
Roughly, a {\em predictor solution} is patched together from the periodic solution in each ``element'', and then a {\em corrector solution} due to the defect is superposed. The efficiency and accuracy of this approach requires that the corrector solution can be coarsely resolved away from the defect \cite{gavini-ortiz-bhatta}.
However, in general there is a spatially-varying dipole moment in the specimen and zero dipole in the free space; therefore, the periodic calculation in any element will have large errors because it replaces this complex environment by a charge distribution with uniform dipole density.
Consequently, the corrector can require fine resolution over much of the domain, except in settings such as metallic systems with net zero local dipole where they are currently applied.

An alternate approach to accounting for the large number of charge-charge interactions is to rewrite the problem as the electrostatic Poisson equation.
However, this will lead to a highly-oscillatory forcing term that fluctuates at the atomic lengthscale while the problem is posed over the entire specimen.
Therefore, this does not solve the essential difficulty.
While numerical homogenization approaches may be feasible because the forcing close-to-periodic in many regions of the sample, it is not clear how to obtain full resolution in the vicinity of the defect.
Further, the Poisson equation can be thought of as a nonlocal constraint that must be appended to (\ref{eq:energy_lattice}) in the minimization and therefore the essential non-local character remains.
%%%%%%%%%%%%%%%%%%%%%
%%%%%%%%%%%%%%%%%%%%%
%%%%%%%%%%%%%%%%%%%%%
%%%%%%%%%%%%%%%%%%%%%

\subsection{Coarse-Graining the Electrostatic Field Energy}
\label{sec:energy_derivation}

We now consider the coarse-graining of the electrostatic energy.
As noted above, we will work with a charge density field $\rho$ and assume that the coarse-graining is also valid for point charges by replacing $\rho$ with appropriate Dirac masses.
Our starting point is to write $\rho$ following the ideas of 2-scale methods \cite{allaire-2scale,debotton-bhattacharya}.
I.e., we consider the setting where the charge density varies over 2 different lengthscales.
There is a rapid almost-periodic variation of charge density at the lengthscale of the atomic unit cell (denoted $l$).
In addition, there is a much slower variation over the characteristic continuum lengthscale denoted $L$.
In the language of 2-scale methods, we can write the charge density field as $\rho(\bfx,\bfy)$ with $\bfy:=\bfx / l$ and $\rho$ periodic (with period of order one) in the second argument.
A heuristic picture is that $\bfx$ specifies the location of the material point, and $\bfy$ specifies the location within the material point (Fig. \ref{fig:unit_cell_omega}).

The electrostatic field energy can be written:
\begin{equation}
\label{eqn:full-field-energy}
	E
		= \int_{\bfx,\bfx' \in\Omega} \frac{\rho(\bfx) \rho(\bfx')}{\left| \bfx - \bfx' \right|} dV_{\bfx} dV_{\bfx'}
\end{equation}
We wish to examine the limit of the energy in the following setting.
We introduce a lengthscale $\epsilon$ that, roughly speaking, denotes the size of the continuum material point.
The limit of interest is then $l / \epsilon \rightarrow 0$ and $\epsilon / L \rightarrow 0$, or $l \ll \epsilon \ll L$.
Essentially, the physical interpretation of this limit is that the atomic unit cell is much smaller than a material point, and a material point is much smaller than the lengthscale over which continuum fields vary.

We can now rewrite $E$ as
\begin{equation}
\label{eqn:full-field-energy-2}
	E
		= \sum_{\bfx,\bfx' \in \Omega} \int_{(l \bfy) \in \mathcal{B}_\epsilon(\bfx), (l \bfy') \in\mathcal{B}_\epsilon(\bfx')} \frac{\rho(\bfx,\bfy) \rho(\bfx',\bfy')}{\left| \bfx + l \bfy - \bfx' - l\bfy' \right|} (l^3 \ dV_{\bfy}) (l^3 \ dV_{\bfy'})
\end{equation}
The notation $\mathcal{B}_\epsilon(\bfx)$ denotes a ball of radius $\epsilon$ centered at $\bfx$.
We note that $(l \bfy) \in \mathcal{B}_\epsilon(\bfx)$ implies $\bfy \in \mathcal{B}_{\epsilon/l}(\bfx)$.

We break up $E$ into 2 parts: a local term when $\bfx=\bfx'$, and a nonlocal term when $\bfx \neq \bfx'$:
\begin{IEEEeqnarray}{rCl}
\label{eqn:full-field-energy-3}
	E
	& = &
	\underbrace{
		\sum_{\bfx \in \Omega} \int_{\bfy, \bfy' \in\mathcal{B}_{\epsilon/l}(\bfx)} \frac{\rho(\bfx,\bfy) \rho(\bfx,\bfy')}{l \left| \bfy - \bfy' \right|} l^6 \ dV_{\bfy} dV_{\bfy'}
	}_{\text{local term: } \bfx = \bfx'} \nonumber \\
	&& +
	\underbrace{
		\sum_{\bfx,\bfx' \in \Omega; \bfx \neq \bfx'} \int_{\bfy\in\mathcal{B}_{\epsilon/l}(\bfx), \bfy' \in\mathcal{B}_{\epsilon/l}(\bfx')} \frac{\rho(\bfx,\bfy) \rho(\bfx',\bfy')}{\left| \bfx + l \bfy - \bfx' - l\bfy' \right|} l^6 \ dV_{\bfy} dV_{\bfy'}
	}_{\text{nonlocal term: } \bfx \neq \bfx'}
\end{IEEEeqnarray}
In the limit that we will take, the nonlocal term represents the interactions between charges that are located at different material points $\bfx, \bfx'$, while the local term represents interactions between charges at the same material point.

We introduce some notation for what follows.
We denote by $\square$ the rescaled atomic unit cell with characteristic dimension and volume of order $1$.
The atomic unit cell with characteristic dimensions $l$ is denoted by $l\square$.
We also use $\square_i$ and $l\square_i$ to denote the $i$-th atomic unit cell in a lattice.

%%%%%%%%%%%%%%%%%%%%%
%%%%%%%%%%%%%%%%%%%%%
%%%%%%%%%%%%%%%%%%%%%
%%%%%%%%%%%%%%%%%%%%%

\subsubsection{The Local Contribution of the Electrostatic Energy}

For a fixed $\bfx$, the charge $\rho(\bfx,\bfy)$ is periodic in the second argument over $\square$.
Therefore, we begin by rewriting the local term in (\ref{eqn:full-field-energy-3}) in terms of integrals over $\square_i$:
\begin{equation}
\label{eqn:local-energy-1}
	\sum_{\bfx \in \Omega} \quad \sum_{\square_i, \square_{i'} \in \mathcal{B}_{\epsilon/l}(\bfx)} \int_{\bfy \in \square_i, \bfy' \in \square_{i'}} l^5 \frac{\rho(\bfx,\bfy) \rho(\bfx,\bfy')}{\left| \bfy - \bfy' \right|} dV_{\bfy} dV_{\bfy'}
\end{equation}
The periodicity, and the fact that $\epsilon / l \rightarrow \infty$, together imply that every term in the sum relating the interaction between cells $i$ and $i'$ can be mapped to an interaction between cells $0$ and some $i''$.
Therefore, the local term can now be written
\begin{equation}
\label{eqn:local-energy-2}
		\sum_{\bfx \in \Omega} \left( \frac{\epsilon}{l}\right)^3 \int_{\bfy \in \square_0, \bfy' \in \mathcal{B}_{\epsilon/l}(\bfx)} l^5 \frac{\rho(\bfx,\bfy) \rho(\bfx,\bfy')}{ \left| \bfy - \bfy' \right|} dV_{\bfy} dV_{\bfy'}
\end{equation}
The factor $(\epsilon / l)^3$ is the number of terms in the sum that are replaced, obtained from dividing the volume of the ball of radius $\epsilon / l$ by the volume of $\square$.

This has the form of a Riemann sum: with $\epsilon \ll L$, the term $\epsilon^3$ is the volume measure.
\begin{equation}
\label{eqn:local-energy-3}
		\sum_{\bfx \in \Omega} \epsilon^3 \left(\int_{\bfy \in \square_0, \bfy' \in \mathcal{B}_{\epsilon/l}(\bfx)} l^2 \frac{\rho(\bfx,\bfy) \rho(\bfx,\bfy')}{\left| \bfy - \bfy' \right|} dV_{\bfy} dV_{\bfy'}\right)
\end{equation}
The term in the brackets is the integrand and must be well-behaved, i.e. neither blow up nor go to $0$, in the limit $l \ll \epsilon$.
The natural scaling is that the charge density must scale as $\rho(\bfx, \bfy) = \tilde{\rho}(\bfX,\bfy) / l$ where $\tilde{\rho}$ is the charge density on the rescaled unit cell $\square$ with characteristic dimension $1$, and $l \bfX = \bfx$.
Note that $\tilde{\rho}$ has dimensions of charge per unit area.
While this choice of scaling may appear arbitrary, we note that it can be recognized as the classical dipole scaling from elementary electrostatics.
That is, in constructing the notion of a point dipole, one starts with charges that are separated by a finite distance and then takes the limit of the charges approaching each other.
However, this limit leads to a finite dipole moment only when the charge magnitude is assumed to scale inversely with separation, thereby leaving the product of charge and separation distance finite.
It is precisely this scaling which is required here for a finite local electrostatic energy.
In our setting, if, for example, we assumed a fixed charge density and allowed the lattice spacing to go to $0$, charge neutrality would give us vanishing energy.

Using the charge scaling described above enables us to map the calculation of the integrand to a unit domain and gives the final form:
\begin{equation}
\label{eqn:local-energy-4}
	\int_{\bfx \in \Omega} \left(\int_{\bfy \in \square_0, \bfy' \in \mathcal{B}_{\epsilon/l}(\bfx)} \frac{\tilde{\rho}(\bfX,\bfy) \tilde{\rho}(\bfX,\bfy')}{\left| \bfy - \bfy' \right|} dV_{\bfy} dV_{\bfy'}\right) \ dV_{\bfx}
\end{equation}
The energy in this form can be readily absorbed into standard energy densities that arise from applying the CB theorem to short-range interactions.
This term has a number of different names: the Madelung energy in ionic solids \cite{kittel-book}, the Lorentz local field, the weak-short contribution \cite{james-muller}.

As an example, we replace the charge density with a set of Dirac masses representing point charges.
The charge density in the unit cell is $\tilde{\rho}(\bfX,\bfy)= \sum_{\bfy\in\square_0} Q^s \delta_{\bfy_s}(\bfy)$ and extended periodically.
The local energy has the form:
\begin{equation}
	\label{eq:local-energy-5}
	E_{local}  =  \int_{\Omega} \left(\sum_{i,j \in \square_0} Q^i \bbS^{ij}Q^j + \frac{1}{3}\vert \bfp(\bfx)\vert^2 +\bfp(\bfx)\cdot \bfS \bfp(\bfx) \right) \ dV_{\bfx}
\end{equation}
where $\bfp(\bfx) := \sum_{\bfy\in\square_0} Q^s \bfy \delta_{\bfy_s}(\bfy)$ is the polarization of the unit cell.

The quantities $\bbS, \bfS$ are defined as:
\begin{equation}
	\label{eq:SS_def}
	\bbS^{ij} := \lim_{\omega \to \infty}\sum_{\bfy^i \in \square}\sum_{\substack{\bfz \in L_1 \setminus \bfy^i \\\cap \mathcal{B}_\omega}} \frac{1}{4 \pi \epsilon_0 \vert \bfy^i-\bfz^j \vert}
	, \quad
	\bfS := \lim_{\omega \to \infty}\sum_{\substack{\bfz \in L_1 \setminus \bfzero \\\cap \mathcal{B}_\omega}} \mathbb{K} (\bfz)
\end{equation}
In the specific case that we have only 2 point charges in a unit cell, $\bbS$ vanishes and the local energy can be written in terms of $\bfp$ exclusively.
The dipole kernel $\mathbb{K}$ is defined in (\ref{eqn:multipole-expansion}).

An important point above is the presence of the limit in the definitions of $\bbS$ and $\bfS$.
As noted previously, the full sums used above are conditionally convergent.
The use of a limit is equivalent to enforcing a particular order of summation; in this case, it corresponds to using ``neutral spheres'' using the terminology of Ewald summation.
Physically, it enforces that the far-field boundary conditions are set to $0$.
The local contribution then is simply the energy of a uniform lattice of charges with vanishing far-field electric field; the lattice is uniform because the entire lattice is located at a single material point.
In general, there can be a non-vanishing far-field electric field due to the other material points and continuum-scale boundary conditions, and this is introduced through the non-local contribution in the next section.

%%%%%%%%%%%%%%%%%%%%%
%%%%%%%%%%%%%%%%%%%%%
%%%%%%%%%%%%%%%%%%%%%
%%%%%%%%%%%%%%%%%%%%%

\subsubsection{Nonlocal Contribution of the Electrostatic Energy}

We now focus on the nonlocal term in (\ref{eqn:full-field-energy-3}), i.e., the interactions between charges at different material points.
This contribution provides an energy that is very different from the standard local continuum energies.
In particular, those energies are developed from the CB theorem that in the limit does not have any direct atomic interactions between different material points.
Here, we have a clear nonlocal character to the energy.

We first introduce some notation regarding the multipole expansion.
Consider the electrostatic interaction for charges located at $\bfx + l \bfy$ and $\bfx' + l\bfy'$:
\begin{equation}
\label{eqn:multipole-expansion}
\begin{split}
 	\frac{1}{\left| \bfx + l\bfy - \bfx' - l\bfy' \right|}
	=
	& \frac{1}{\left| \bfx - \bfx' \right|}
	+
	\frac{\partial }{\partial \left(\bfx - \bfx' \right)} \left(\frac{1}{\left| \bfx - \bfx' \right|}\right) l \cdot \left( \bfy - \bfy' \right)
	\\
	&
	+
	\half
	\underbrace{
		\frac{\partial^2 }{\partial \left(\bfx - \bfx' \right)^2} \left(\frac{1}{\left| \bfx - \bfx' \right|}\right)
	}_{\mathbb{K}} l^2 : \left( \bfy - \bfy' \right) \otimes \left( \bfy - \bfy' \right)
	+
	O(l^3)
\end{split}
\end{equation}
The operator $\mathbb{K}$ is the dipole kernel.

In the nonlocal term in (\ref{eqn:full-field-energy-3}), in anticipation of using the periodicity of $\rho(\bfx,\bfy)$ in $\bfy$ when $\bfx$ is held fixed, we reduce the integrations to unit cells $\square$:
\begin{equation}
\label{eq:nonlocal-energy-1}
	\sum_{\bfx,\bfx' \in \Omega; \bfx \neq \bfx'}
	\quad
	\sum_{\square_i \in \mathcal{B}_{\epsilon/l}(\bfx); \square_{i'} \in \mathcal{B}_{\epsilon/l}(\bfx')}
	\quad
	\int_{\bfy \in \square_i, \bfy' \in \square_{i'}} \frac{\rho(\bfx,\bfy) \rho(\bfx',\bfy')}{\left| \bfx + l \bfy - \bfx' - l\bfy' \right|} l^6 \ dV_{\bfy} dV_{\bfy'}
\end{equation}
Assuming a separation of scales, i.e. $\epsilon \ll L$, the periodicity of $\rho$ implies that the interaction between charges contained in $\mathcal{B}_{\epsilon/l}(\bfx)$ and $\mathcal{B}_{\epsilon/l}(\bfx')$ can be replaced by interactions between charges in unit cells at $\bfx$ and $\bfx'$, and then multiplying by the number of unit cells in  $\mathcal{B}_{\epsilon/l}(\bfx)$ and $\mathcal{B}_{\epsilon/l}(\bfx')$.
\begin{equation}
\label{eq:nonlocal-energy-2}
	\sum_{\bfx,\bfx' \in \Omega; \bfx \neq \bfx'}
	\quad
	\left(\frac{\epsilon}{l}\right)^3
	\left(\frac{\epsilon}{l}\right)^3
	\int_{\bfy \in \square, \bfy' \in \square} \frac{\rho(\bfx,\bfy) \rho(\bfx',\bfy')}{\left| \bfx + l \bfy - \bfx' - l\bfy' \right|} l^6 \ dV_{\bfy} dV_{\bfy'}
\end{equation}
Canceling the factors of $l$ and using the notion of Riemann sums as above with $\epsilon^3$ as the volume measure, we can write this as a double integral:
\begin{equation}
\label{eq:nonlocal-energy-3}
	\int_{\bfx,\bfx' \in \Omega; \bfx \neq \bfx'}
	\left(
	\int_{\bfy \in \square, \bfy' \in \square} \frac{\rho(\bfx,\bfy) \rho(\bfx',\bfy')}{\left| \bfx + l \bfy - \bfx' - l\bfy' \right|} \ dV_{\bfy} dV_{\bfy'}
	\right)
	\ dV_{\bfx} \ dV_{\bfx'}
\end{equation}
As in the local contribution, we require the integrand in brackets above to be well-defined when $l \rightarrow 0$.
Recall the dipole scaling $\rho = \tilde{\rho} / l$: the integrand is therefore well-behaved if $\frac{1}{\left| \bfx + l \bfy - \bfx' - l\bfy' \right|}$ scales as $l^2$.

We substitute the multipole expansion from (\ref{eqn:multipole-expansion}) and notice immediately that the first term scales independently of $l$ and the second term scales linearly in $l$.
These would then potentially cause the integrand to diverge as $l\rightarrow 0$.
However, we recall that each unit cell is charge-neutral, i.e. $\int_\square \tilde{\rho}(\bfx,\bfy) \ dV_{\bfy} = 0$.
This causes both the first and second terms in the multipole expansion to vanish.
Physically, this means that the energy is unbounded if every unit cell is not charge-neutral, e.g. recalling the example in Section \ref{toy-example}.

Next, we consider the term $\half l^2 \mathbb{K} : \left( \bfy - \bfy' \right) \otimes \left( \bfy - \bfy' \right)$.
The terms containing $\bfy \otimes \bfy$ and $\bfy' \otimes \bfy'$ vanish from charge neutrality.
The only remaining terms can be readily written as:
\begin{equation}
\label{eq:nonlocal-energy-4}
\begin{split}
	\int_{\bfx,\bfx' \in \Omega; \bfx \neq \bfx'}
	\left(
	\mathbb{K}(\bfx-\bfx') :
	\underbrace{
		\int_{\bfy \in \square} \tilde{\rho}(\bfX,\bfy) \bfy \ dV_{\bfy}
	}_{\bfp(\bfx)}
	\otimes
	\underbrace{
		\int_{\bfy' \in \square} \tilde{\rho}(\bfX',\bfy') \bfy' \ dV_{\bfy'}
	}_{\bfp(\bfx')}
	\right)
	\ dV_{\bfx} \ dV_{\bfx'}
	\\
	=
	\int_{\bfx,\bfx' \in \Omega; \bfx \neq \bfx'} \bfp(\bfx') \cdot \mathbb{K}(\bfx-\bfx') \bfp(\bfx) \ dV_{\bfx} \ dV_{\bfx'}
\end{split}	
\end{equation}
where the terms containing $\bfy \otimes \bfy'$ and $\bfy' \otimes \bfy$ have been combined using symmetry.

Consider finally the terms denoted $O(l^3)$.
These will all go to $0$ as $l\rightarrow 0$.
Physically, these terms represent the contributions from quadrupole and higher-order moments of the charge distribution, i.e. $\int_\square \tilde{\rho}(\bfx,\bfy) \bfy \otimes \bfy \ dV_{\bfy}$ and higher-order.
We see that these terms vanish identically in the limit.
Therefore, terms of higher-order than dipole do not appear in the nonlocal part of the continuum energy, recalling the example in Section \ref{toy-example}.
In general, all short-range forces -- i.e. those that decay faster than dipolar interactions -- do not contribute to the nonlocal term in the limit \cite{friesecke-james, blanc-lebris-lions}.

Finally, a long but straightforward calculation using the divergence theorem and integration-by-parts gives
\begin{equation}
\label{eq:nonlocal-energy-5}
\begin{split}
	& \int_{\bfx,\bfx' \in \Omega} \bfp(\bfx') \cdot \mathbb{K}(\bfx-\bfx') \bfp(\bfx) \ dV_{\bfx} \ dV_{\bfx'} 
	=
	\\
	& \int_{\bfx,\bfx' \in \Omega} \divergence \bfp(\bfx') G(\bfx-\bfx') \divergence \bfp(\bfx) \ dV_{\bfx} \ dV_{\bfx'}
	+ \int_{\bfx,\bfx' \in \partial\Omega} \bfn \cdot \bfp(\bfx') G(\bfx-\bfx') \bfn \cdot \bfp(\bfx) \ dS_{\bfx} \ dS_{\bfx'} 
	\\
	& - 2 \int_{\bfx \in \Omega,\bfx' \in \partial\Omega} \bfn \cdot \bfp(\bfx') G(\bfx-\bfx') \divergence \bfp(\bfx) \ dV_{\bfx} \ dS_{\bfx'} 	
\end{split}
\end{equation}
where $G$ is the standard electrostatics Greens function and $\mathbb{K}:=\nabla^2 G$ from (\ref{eqn:multipole-expansion}).
Note that the condition $\bfx\neq\bfx'$ has not been written for brevity.
An important conclusion from the above formula is that $- \divergence\bfp$ is equivalent to a bulk charge density (the so-called ``bound bulk charge density'') and that $\bfp\cdot\bfn$ is equivalent to a surface charge density (the so-called ``bound surface charge density'').
In that perspective, the formula above simply gives the energy of this composite charge distribution using the usual Green's function relation between charge density and energy.
When $\bfp$ is discontinuous along interior surfaces, this formula gives bound surface charges along these surfaces.

While we have derived the equivalent bound charges using standard integration formulas after taking the limit of the Riemann sum, it is straightforward to derive these directly.
Essentially, we manipulate the Riemann sum using the standard approach in the Riemann sum proof of the divergence theorem (e.g., \cite{tang-book}).
In the interior, we find $\divergence\bfp$ appearing and the boundaries of the infinitesimal element canceling with it's neighbors.
On the boundary of $\partial\Omega$, there is no cancellation leading to the contribution $-\bfp\cdot\bfn$.

%%%%%%%%%%%%%%%%%%%%%
%%%%%%%%%%%%%%%%%%%%%
%%%%%%%%%%%%%%%%%%%%%
%%%%%%%%%%%%%%%%%%%%%

\subsubsection{Boundary Contributions due to Partial Unit Cells}
\label{sec:boundary-charge}

We consider now the role of boundaries.
As we have mentioned above and will discuss below in detail, the value of the polarization depends on the chosen unit cell.
Boundary contributions are essential to ensure that the coarse-grained electric fields and other quantities do not depend on the arbitrary choice of unit cell.
There are two contributions: first, due to polarization terminations $-\bfp\cdot\bfn$, and second, surface charges due to partial unit cells on the surface that are not charge-neutral.
The polarization terminations have already been accounted for as shown in (\ref{eq:nonlocal-energy-5}).
In this section, we consider the case of the surface charges due to partial non-charge-neutral unit cells.

For simplicity, we do not compute the total energy which will have straightforward but tedious cross-terms between interior bound charges (due to $\divergence \bfp$) and surface charges as can be seen (\ref{eq:nonlocal-energy-5}).
Instead, we compute the electric potential field due to the surface charges where the calculations are more transparent.
In a formal setting, one can readily go between these calculations.

Consider a point $\bfx \in \partial\Omega$.
At a point $\bfx' \neq \bfx$, the electric potential due to the charges at $\bfx$ is given by:
\begin{equation}
\label{eqn:boundaries-1}
	\phi(\bfx') = \sum_{\bfx \in \partial\Omega} \int_{l\bfy \in \mathcal{D}_{\epsilon}(\bfx) \times C l \bfn} \frac{\rho(\bfx,\bfy)}{|\bfx + l\bfy - \bfx'|} l^3 \ dV_{\bfy}
\end{equation}
Here $\mathcal{D}_{\epsilon}(\bfx)$ is a 2D disk of radius $\epsilon$ located at $\bfx$.
Therefore, $\mathcal{D}_{\epsilon}(\bfx) \times l \bfn$ denotes a squat cylinder of height $C l$ oriented with axis $\bfn$ and cross section $\mathcal{D}_{\epsilon}(\bfx)$.
The vector $\bfn$ is the unit outward normal to $\Omega$.
It is implicit in the above formula that we are considering charges only in the partial unit cells.

We assume that the surface is a rational plane (see Appendix \ref{sec:appendix} for the definition).
The integration above in the directions along the surface (i.e., normal to $\bfn$) can then be reduced to an integration over a single unit cell because of periodicity in those directions\footnote{
It is not clear to us how to proceed without assuming that the surfaces are rational planes.  Irrational surfaces cause severe difficulties in defining surface energies even in simpler models of solids \cite{rosakis-surface-energy}.
}.
The integration in the direction along $\bfn$ reduces simply to the partial unit cells on the boundary and is therefore independent of $C l$.
Following the ideas above for the volume contributions, we can then rewrite this as an integral over a unit cell by putting in the appropriate factor for the number of unit cells in the disk:
\begin{equation}
\label{eqn:boundaries-2}
	\phi(\bfx') = \sum_{\bfx \in \partial\Omega} \frac{\epsilon^2}{l^2} \int_{\bfy \in \triangle} \frac{\tilde{\rho}(\bfX,\bfy)/l}{|\bfx + l\bfy - \bfx'|} l^3 \ dV_{\bfy}
\end{equation}
where we have used the notation $\triangle$ for partial non-neutral unit cells.

Therefore, we require that only the term independent of $l$ from (\ref{eqn:multipole-expansion}) appear above.
Upon taking the Riemann sum and defining the surface charge density $\sigma(\bfx)$, this gives the expected and simple result:
\begin{equation}
\label{eqn:boundaries-3}
	\phi(\bfx') = \int_{\bfx \in \partial\Omega} \frac{1}{|\bfx - \bfx'|} \underbrace{\left( \int_{\bfy \in \triangle} \tilde{\rho}(\bfX,\bfy) \ dV_{\bfy} \right)}_{\sigma} \ dS_{\bfx}
\end{equation}
While we have assumed for simplicity that $\bfx \neq \bfx'$, considering the case $\bfx = \bfx'$ and examining the energy would give us a local contribution analogous to that in the case of the bulk.

%%%%%%%%%%%%%%%%%%%%%
%%%%%%%%%%%%%%%%%%%%%
%%%%%%%%%%%%%%%%%%%%%
%%%%%%%%%%%%%%%%%%%%%

\subsection{Role of Boundaries in Compensating for the Non-uniqueness of Polarization}

It is well-known, e.g. \cite{Resta-Vanderbilt}, that the value of $\bfp$ in a periodic solid depends on the choice of unit cell.
This would appear to be a fatal difficulty in using $\bfp$ as a multiscale mediator between the atomic-scale-variation of $\rho$ and continuum-scale quantities.
In the materials physics community, quantum mechanical notions are invoked to obtain a unique choice of polarization for a given periodic charge distribution \cite{Resta-Vanderbilt}.
However, from the perspective of the calculations above, we are simply coarse-graining classical electrostatic interactions and there is no reason for quantum mechanics to play any role.
As we describe in this section, the difficulties noticed by \cite{Resta-Vanderbilt} are entirely due to their starting-point of an infinite periodic solid.
This makes the notion of boundaries ill-defined.
If instead we begin from a finite solid and take the limit of lattice spacing being much smaller than the size of the body (the large-body limit), we see that the surface charges on the boundaries play a critical role.
In short, while the polarization is itself not uniquely-defined, the electrostatic energy that comes from accounting for both the polarization and the surface charge is a unique quantity.
While changing the unit cell changes the value of the polarization density, it also changes the boundary charge in the partial unit cells.
These compensate to give the same value for the electrostatic energy.

%%%%%%%%%%%%%%%%%%%%%
%%%%%%%%%%%%%%%%%%%%%
%%%%%%%%%%%%%%%%%%%%%
%%%%%%%%%%%%%%%%%%%%%

\subsubsection{One-Dimensional Illustrative Example}
\label{sec:unit_cell_example}

Consider a finite body $\Omega = (-L,L)\times (-1,1) \times (-1,1)$ with a one-dimensional charge distribution $\rho(\bfx) = \rho_0 \sin (2\pi \frac{x_1}{l})$ (Fig. \ref{fig:1D-example-boundaries}).
We assume that $L$ is an integer multiple of $l$, i.e. $n l = L$.
Guided by the multipole expansion, we compute the dipole moment as the leading contributor to the behavior of the bar without using the fact that the charge distribution is in fact periodic in $\Omega$.
We then compare this to the result obtained by using polarization density field that is defined on the unit cell.

\begin{figure}[ht!]
	\centering
	\includegraphics[width=\textwidth]{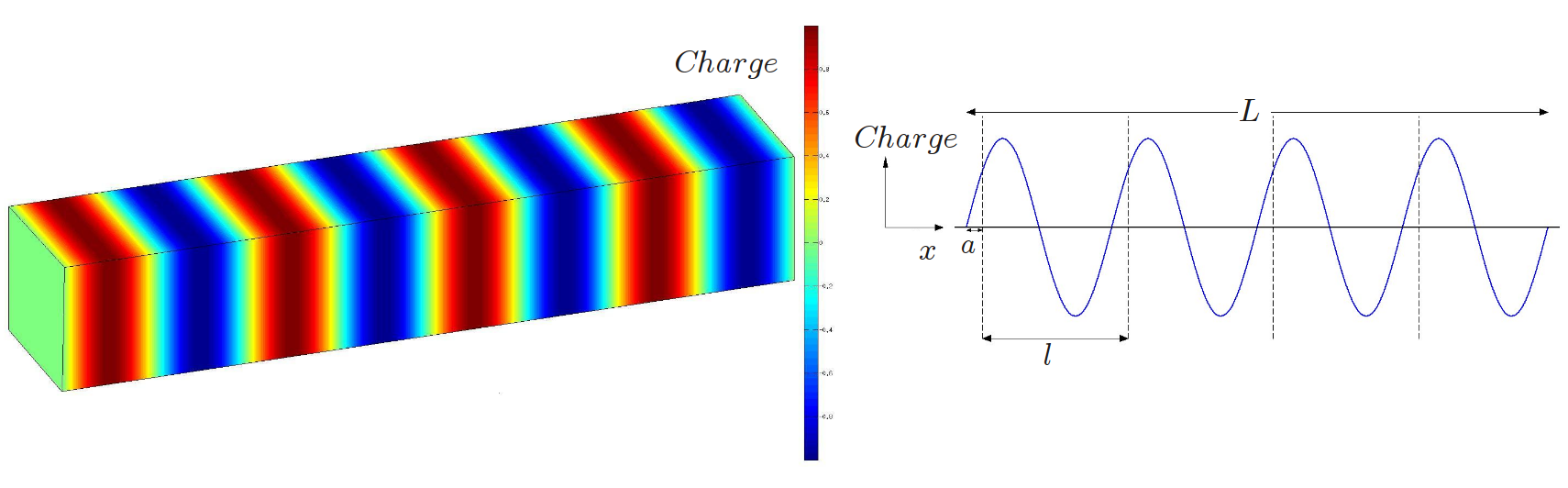}
	\caption{\small Charge distribution in the 1D illustrative example to show the effect of boundaries.}
	\label{fig:1D-example-boundaries}
\end{figure}

The total dipole moment of the bar $\bfP := \int_{\Omega}  \rho(\bfr') \bfr' \ dV_{\bfr'}$ evaluates to $(1,0,0) \times \frac{-4\rho_0 L l}{ \pi}$.

The polarization density $\bfp$ in a single unit cell $l\square = (a,l+a)\times (-1,1) \times (-1,1)$ is defined as $\bfp := \frac{1}{\text{volume}(l \square)} \int_{\bfr' \in l \square}  \rho(\bfr') \bfr' \ dV_{\bfr'}$.
It evaluates to $(1,0,0) \times \frac{- \rho_0 l}{2 \pi} \cos(2\pi\frac{a}{l})$.
This is a classic example showing that $\bfp$ depends on the chosen unit cell, here parametrized by $a$ \cite{Resta-Vanderbilt}.
From (\ref{eq:nonlocal-energy-5}) and the associated discussion, we have no bulk charge because $\bfp$ is identical in each unit cell, but there is a surface charge density given by $-\bfp\cdot\bfn$.
Therefore, there is a total charge of $\frac{2 \rho_0 l}{\pi} \cos(2\pi\frac{a}{l})$ at $+L$, and $\frac{- 2 \rho_0 l}{\pi} \cos(2\pi\frac{a}{l})$ at $-L$.
However, there are partial unit cells at each end: $(-L,-L+a)$ and $(L-l+a, L)$, and these are not charge neutral.
The charge in these cells evaluates to $\frac{-2 \rho_0 l}{\pi} \left( 1 - \cos(2\pi\frac{a}{l})\right)$ and $\frac{2 \rho_0 l}{\pi} \left( 1 - \cos(2\pi\frac{a}{l})\right)$ respectively.
Therefore the total charge at each end, both from the partial unit cells and from $-\bfp\cdot\bfn$, is $\pm \frac{2 \rho_0 l}{\pi}$.
These equal and opposite charges are separated by a distance $2 L$.
Therefore, this is a dipole of strength $\frac{-4 \rho_0 l L}{\pi} $.
Note that we have errors up to order $l$ because the charges due to the polarization terminating on the surface are separated by $L-l$; however, the key point is that in the limit of $l \ll L$, we recover the dipole $\bfP$.

%%%%%%%%%%%%%%%%%%%%%
%%%%%%%%%%%%%%%%%%%%%
%%%%%%%%%%%%%%%%%%%%%
%%%%%%%%%%%%%%%%%%%%%

\subsubsection{The General Case}

The key lesson from the example above is that it is critical to account for the charge in the partial unit cells on the boundary.
This ensures that the coarse-graining that exploits the polarization as a multiscale mediator does not depend on the choice of unit cell.
We now examine this in the 3D setting.

First, we decompose $\Omega$ into $\Omega_\square$, with only complete unit cells, and $\Omega_{\#}$ with the surface layer of incomplete unit cells, Fig. \ref{fig:omega-decomposition}.

\begin{figure}[ht!]
	\centering
	\includegraphics[width=180mm]{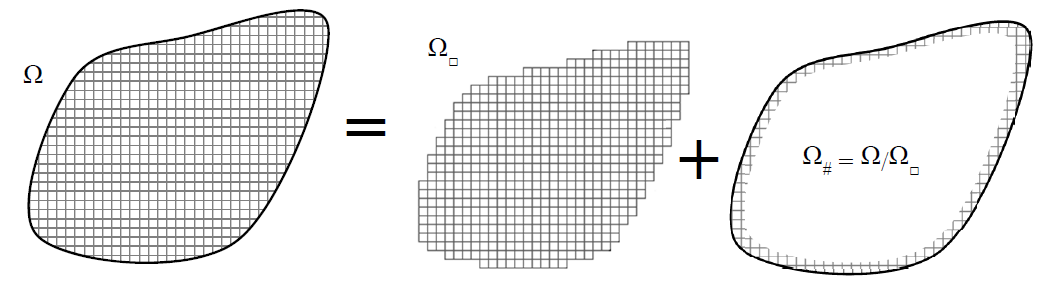}
	\caption{\small Decompose $\Omega$ into $\Omega_\square$ and $\Omega_{\#}$.}
	\label{fig:omega-decomposition}
\end{figure}

We now consider the unit cells adjacent to $\Omega_{\#}$.
Our goal is to modify these unit cells in various ways, and show that the resulting changes in surface charge and polarization density balance each other.

An important element of our strategy is to transform the given unit cell to a unit cell that has a face parallel to the surface under consideration.
Appendix \ref{sec:appendix} shows that a unit cell with this property can always be found when the surface is rational.
As in Section \ref{sec:boundary-charge}, we restrict attention to rational surfaces.

For this special choice of unit cell, we now consider the changes in surface charge and polarization density for various operations.
We use the notation that the lattice vectors tangential to the surface are $\bfh_2$ and $\bfh_3$.
We note that the volume of the unit cell can be written $\bfh_2 \times \bfh_3 \cdot \bfh_1 = \vert \bfh_2 \times \bfh_3 \vert \bfn \cdot \bfh_1$.

First, consider translations of the unit cell as shown in Fig. \ref{fig:unit-cell-change-1}.

Consider Fig. \ref{fig:unit-cell-change-1}a.
A translation along $\bfh_1$ causes an increase in the uncompensated surface charge area density $\Delta\sigma = \frac{1}{\vert \bfh_2 \times \bfh_3 \vert} \int_{\text{\textcircled{1}}} \rho$ in the partial unit cells.
The increase in the polarization density in the translated unit cell is $\Delta\bfp = \frac{1}{\vert \bfh_2 \times \bfh_3 \vert \bfn \cdot \bfh_1} \left( \int_{\text{\textcircled{2}}} \rho \bfy - \int_{\text{\textcircled{1}}} \rho \bfy \right)$.
From the periodicity of the charge distribution $\rho(\bfy+\bfh_1) = \rho(\bfy)$, we have $\Delta\bfp = \frac{1}{\vert \bfh_2 \times \bfh_3 \vert \bfn \cdot \bfh_1} \bfh_1 \int_{\text{\textcircled{2}}} \rho $.
Therefore $\Delta\bfp\cdot\bfn = \Delta\sigma$.

Consider Fig. \ref{fig:unit-cell-change-1}b.
A translation along either $\bfh_2$ or $\bfh_3$ causes no change in $\sigma$.
The change in polarization is $\frac{1}{\vert \bfh_2 \times \bfh_3 \vert \bfn \cdot \bfh_1} \bfh_2 \int_{\text{\textcircled{2}}} \rho $.
Therefore $\Delta\bfp \cdot\bfn = 0$.

In general, one can have a translation along any direction.
Such a translation can be decomposed into components along the lattice directions and the calculations above applied in succession to each direction.

\begin{figure}[ht!]
	\centering
	\includegraphics[width=180mm]{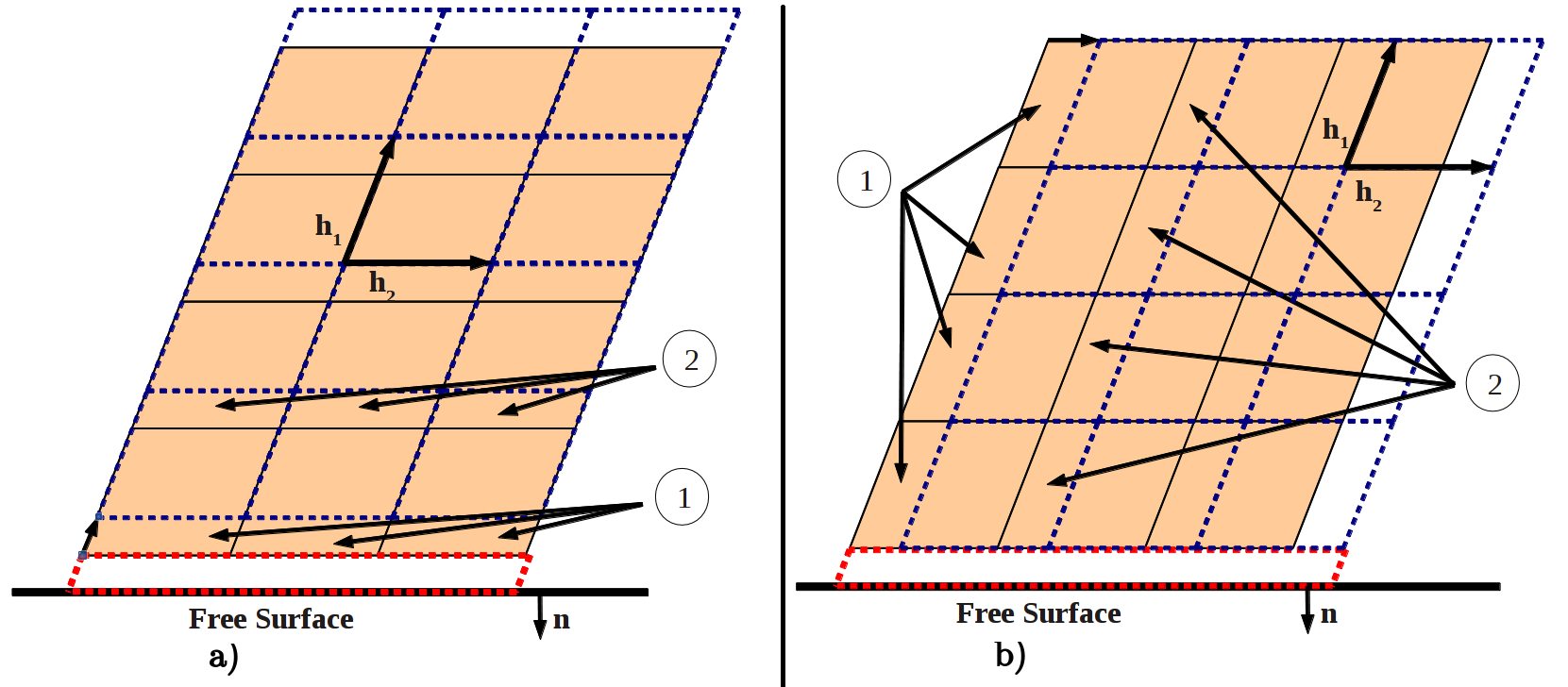}
	\caption{\small Change in charge and polarization due to a translation along $\bfh_1$ and $\bfh_2$ respectively in the special choice of unit cell.}
	\label{fig:unit-cell-change-1}
\end{figure}

Second, consider distortions of the unit cell as shown in Fig. \ref{fig:unit-cell-change-2}.

From reasoning following very closely the previous case of translations of the unit cell, we find that the relation $\Delta\bfp\cdot\bfn = \Delta\sigma$ holds here too.

\begin{figure}[ht!]
	\centering
	\includegraphics[width=180mm]{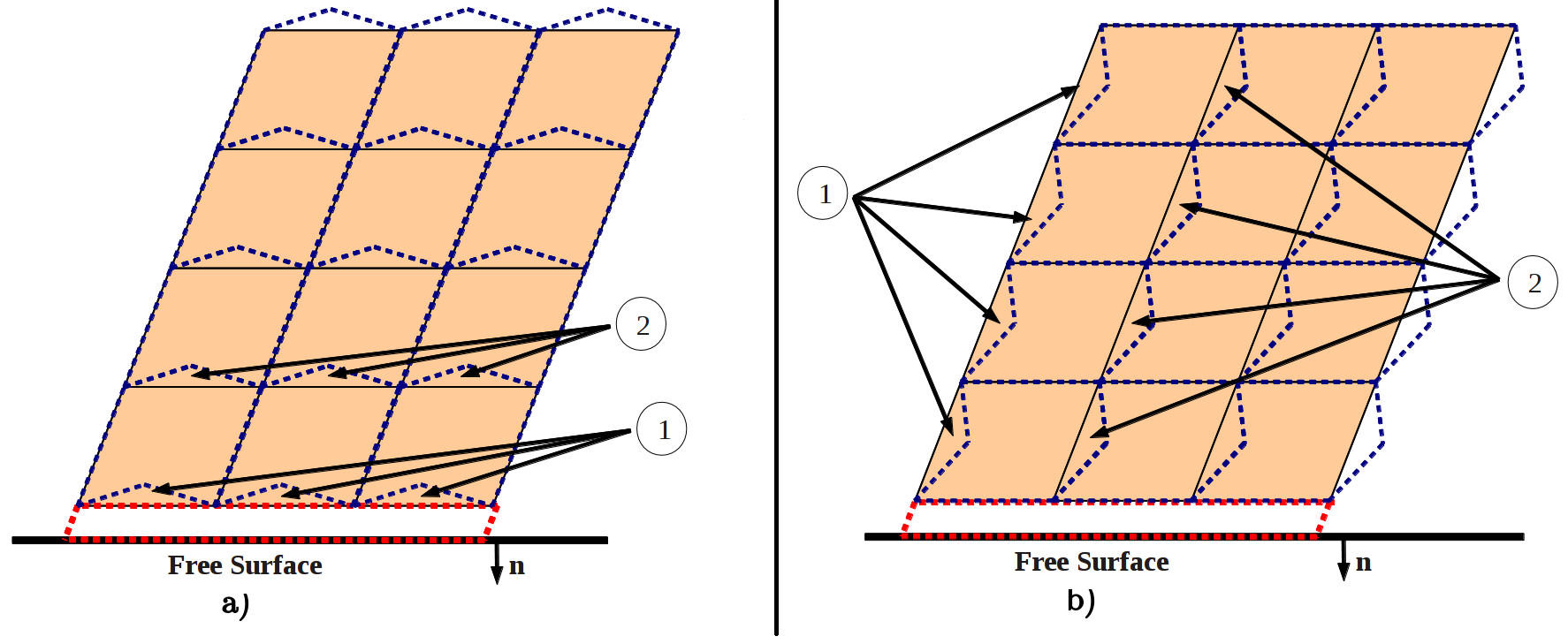}
	\caption{\small Change in charge and polarization due a distortion of the unit cell in the special choice of unit cell.}
	\label{fig:unit-cell-change-2}
\end{figure}

Third, consider changes in the unit cell due to remapping of the lattice vectors as shown in Fig. \ref{fig:unit-cell-change-3}.

As noted in Appendix \ref{sec:appendix}, the relation between $\{ \bff_1, \bff_2, \bff_3 \}$ and $\{ \bfh_1, \bfh_2, \bfh_3 \}$ must be of the form: $\bff_i = \sum_j \mu_i^j \bfh_j$ with $\mu_i^j$ a matrix of integers with determinant $\pm 1$.
An example of such a remapping is in Fig. \ref{fig:unit-cell-change-3}a.

As shown in the example in Fig. \ref{fig:unit-cell-change-2}bcde, regions of the original unit cell are mapped to the new unit cell.
For instance, \circled{1a} maps to \circled{1b}, \circled{2a} maps to \circled{2b}, and \circled{3a} maps to \circled{3b}.
Each of these regions is translated by an integer linear combination of $\{ \bfh_1, \bfh_2, \bfh_3 \}$; however, each region may have a {\em different} integer combination translation.
In addition, the uncompensated unit cell on the boundary also increases in extent (Fig. \ref{fig:unit-cell-change-2}f).
As above, for a periodic charge distribution, when a charged region is translated by an integer multiple of a lattice vector, the consequent change in polarization is simply the total charge in the region times the translation distance.
Therefore $\Delta\bfp = \frac{1}{\vert \bfh_2 \times \bfh_3 \vert \bfn \cdot \bfh_1} \left(\sum_i \nu^1_i \bfh_1 \int_{A_i} \rho + \sum_i \nu^2_i \bfh_2 \int_{A_i} \rho + \sum_i \nu^3_i \bfh_3 \int_{A_i} \rho\right)$, where $\nu^1_i,\nu^2_i,\nu^3_i$ are integers.

The change in the uncompensated charge is simply $\Delta\sigma = \frac{1}{\vert \bfh_2 \times \bfh_3 \vert} \sum_i \nu^1_i \int_{A_i} \rho$ which is related to the extent of the translation along $\bfh_1$.
This matches precisely with $\Delta \bfp\cdot\bfn$.

\begin{figure}[ht!]
	\centering
	\includegraphics[width=180mm]{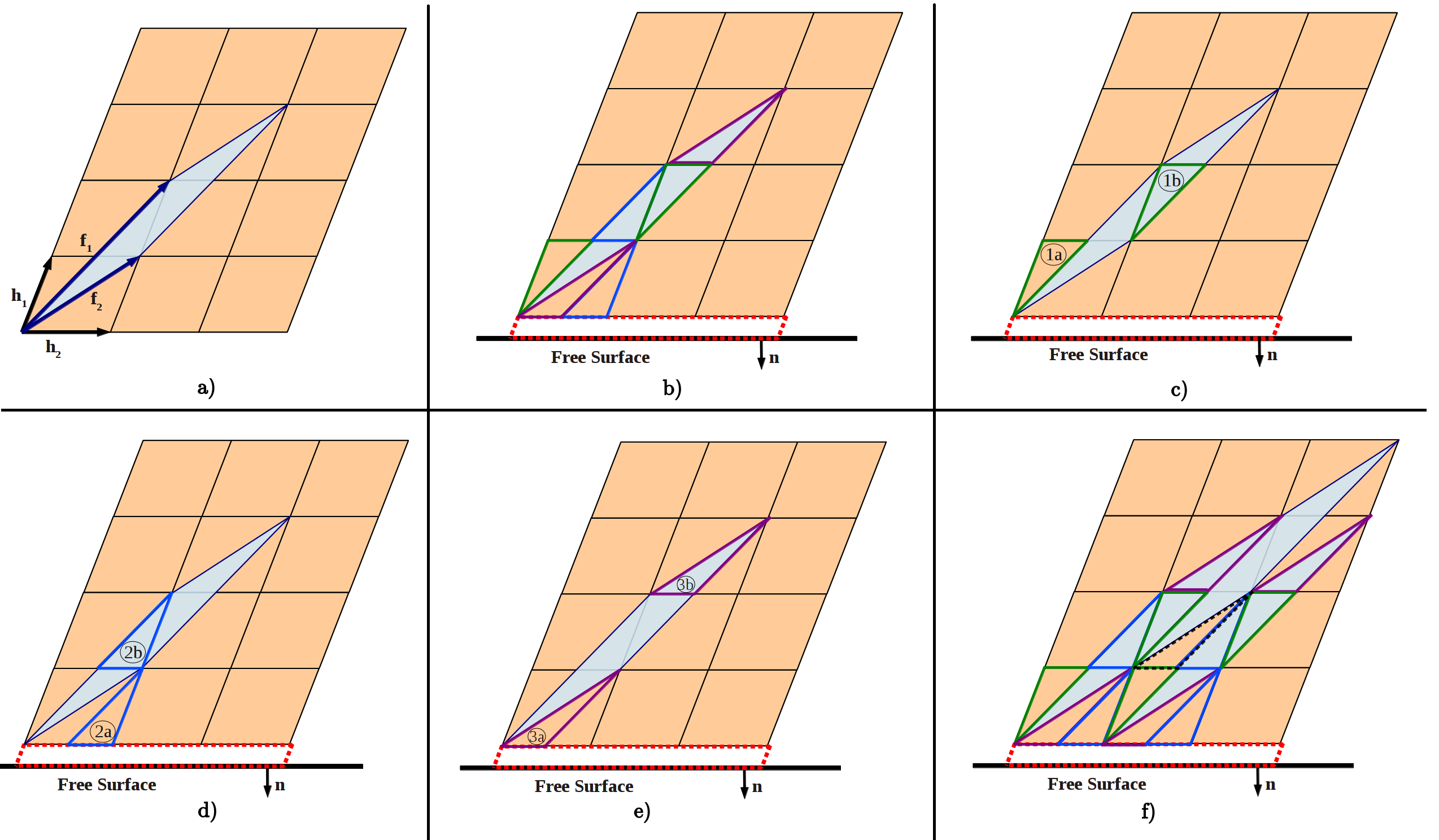}
	\caption{\small Change in charge and polarization due a remapping of the unit cell from the special choice of unit cell.}
	\label{fig:unit-cell-change-3}
\end{figure}

In the general case of modifying a given unit cell to another shape by any combination of the mechanisms studied above, we can conceptually consider mapping the given unit cell to the special unit cell with a surface-parallel face, conducting the modifications with the special unit cell, and then mapping to the desired final unit cell.
Of course, in practice none of this need be done; as long as we are assured that changes in the unit cell are compensated by boundary charges appropriately, we can directly modify the unit cell as desired.

In the interior of the body, $\divergence\bfp$ changes by $O(l)$ when the unit cell is changed.
This follows directly from the definition of $\bfp$ in (\ref{eq:nonlocal-energy-4}) and the chain rule, using the fact that $\bfX / l = \bfx$.
Therefore, the bound bulk charge density is the same in the limit of $l/L \ll 1$.

We note potential connections to ideas of Null-Lagrangians in the issue of a unique definition of the polarization \cite{ericksen-nilpotent}.
Essentially, one can have different expressions for the free-energies, but these lead to the same Euler-Lagrange equation but with different boundary conditions.
Similarly, different choices for the unit cell leave the bulk electrostatics unchanged, because the change in the bulk polarization is compensated by the boundary contributions.

%%%%%%%%%%%%%%%%%%%%%
%%%%%%%%%%%%%%%%%%%%%
%%%%%%%%%%%%%%%%%%%%%
%%%%%%%%%%%%%%%%%%%%%

\section{Numerical Implementation}
\label{sec:implementation}

Our numerical implementation for the short-range interactions follows closely the standard QC, e.g. \cite{Tadmor}.
In standard QC, there are two essential steps: first, a reduction of degrees of freedom by interpolation, typically using linear shape functions inspired by finite elements, with atomic resolution in critical regions and coarse-grained elsewhere (Figs. \ref{fig:mesh}, \ref{fig:mesh_zoom}); and second, a fast estimation of the energy (or derivative of the energy with respect to the retained degrees of freedom) using Cauchy-Born sampling.
Among the many variants of QC, we use the local QC for multi-lattices, following \cite{Tadmor}.
Local QC refers to the use of sampling {\em everywhere} in the specimen, not only in the coarse-grained region but also in regions with atomic resolution in the interpolation.

Dielectrics require a multi-lattice description because a dielectric response requires at least two charges in the unit cell that move independently to change the polarization in response to electromechanical fields.
Essentially, the multi-lattice description uses $\bfF$ to track the deformation of the unit cell, and a set of vector-valued fields $\bfzeta^s$ that track the position of individual species $s$ within the unit cell.
We use linear interpolation for the coarse-grained displacement field $\bfu$ implying that $\bfF$ is constant in a given element.
To be consistent with this spatial variation of $\bfF$, we use the same ``constant in each element'' interpolation for $\bfzeta^s$.
If the element were of infinite extent, this choice of interpolation would ensure that the energy density converges to the standard Cauchy-Born theorem.
In the local QC approximation, we estimate the energy of a given element by finding the energy density of atomic unit cells at the selected quadrature / sampling points and multiply by the appropriate weight.

This interpolation of the kinematic variables $\bfF, \bfzeta^s$ implies that the polarization $\bfp$ is also constant in a given element.
Therefore, in terms of effective bound charges, we have no bound bulk charges ($\divergence \bfp = 0$), and we have surface charge density $(\bfp_1 - \bfp_2)\cdot \bfn$ at the element faces.
In the coarse-grained local QC approximation,  the evaluation of electrostatic fields consists simply of finding the fields set up by the charge distributions.
One element of electrostatics is that the electric potential and field are naturally posed in the current configuration.
Therefore, for numerical updates, we compute the electrostatic fields in the current and then pull back to the reference using standard electromechanical transformations \cite{xiao-bhatta}.
We assume in the remainder of the paper that we are dealing with point charges that can move around, but do not change their charge.

%%%%%%%%%%%%%%%%%%%%%
%%%%%%%%%%%%%%%%%%%%%
%%%%%%%%%%%%%%%%%%%%%
%%%%%%%%%%%%%%%%%%%%%

\subsection{Energy Minimization through Gradient Descent}

Our interest is in finding energy minimizers to the coarse-grained problem.
Under the local QC approximation, the coarse-grained problem can be written as the standard continuum energy for electromechanical solids \cite{shu-bhatta,xiao-bhatta}:
\begin{equation}
\label{eq:total_energy_again}
	E  =  \underbrace{\int_{\Omega_0} W(\grad_{\bfx_0} \mathbf{u},\bfzeta^s) \ dV_{\bfx_0}}_\text{short-range energy}
		+ \underbrace{ \frac{\epsilon_0}{2}\int_{\R^3}\vert \grad_{\bfx} \phi \vert^2 \ dV_{\bfx} }_\text{long-range (non-local) energy}
		 \quad , \quad \divergence_{\bfx} \grad_{\bfx} \phi = \divergence_{\bfx} \bfp
\end{equation}
The nonlocal expression in (\ref{eq:nonlocal-energy-5}) can be transformed to the form above by first noting that the right side of (\ref{eq:nonlocal-energy-5}) is entirely in terms of the electrostatic Greens function $G$.
Therefore, the electrostatic field $\phi$ can be defined as the solution of the electrostatic equation with charges given by $-\divergence\bfp$ and $\bfp\cdot\bfn$.
The energy density is then simply $\vert \grad_{\bfx} \phi \vert^2$.
The surface charges due to $\bfp\cdot\bfn$ as well as due to incomplete unit cells appear in the boundary conditions for the electrostatic equation.
The local contribution of the electrostatic energy has been absorbed into the short-range term.
Note that the non-local energy integral is posed in the current configuration.

We use a gradient-flow evolution to find the (local) minimizers following \cite{xiao-bhatta,zhang-bhatta}.
The independent variables that remain are $\bfu$ and $\bfzeta^s$.
The polarization is completely defined in terms of unit cell geometry $\bfu$ and intra-cell positions of the charges $\bfzeta^s$, in turn defining the electrostatic potential $\phi$ through the electrostatic Poisson equation.
Therefore, taking variations:
\begin{equation}
\label{eq:deformation_variation}
	\grad_{\bfx_0} \bfu \to \grad_{\bfx_0} \bfu +\eta \grad_{\bfx_0} \bfv
	\quad , \quad
	\bfzeta^s \to \bfzeta^s + \eta \bftheta^s
\end{equation}
Additionally, the variation in the electric field is $\grad_{\bfx} \phi \to \grad_{\bfx} \phi +\eta \grad_{\bfx} \psi$ but this variation is constrained to variations in $\bfp \to \bfp + \eta \bfq$ by the Poisson equation, i.e. $\divergence_{\bfx} \grad_{\bfx} \psi = \divergence_{\bfx} \bfq$.

The definition of gradient flow results in the following statement:
\begin{equation}
\label{eq:grad_flow}
\begin{split}
	\int_{\Omega_0} \left (\dot{\bfu}\cdot \bfv + \sum_{s} \dot{\bfzeta}^s \cdot \bftheta^s \right ) \ dV_{\bfx_0}
		 = & - \left. \frac{\partial{}}{\partial{\eta}} E \left ( \grad_{\bfx_0} \bfu +\eta \grad_{\bfx_0} \bfv, \bfzeta^s + \eta  \bftheta^s \right) \right|_{\eta=0} \\
		 = & \int_{\Omega} \left( \frac{\partial{W}}{\partial{(\grad_{\bfx_0} \bfu)}}: \grad_{\bfx_0}\bfv + \sum_{s} \frac{\partial{W}}{\partial{\boldsymbol{\zeta^s}}} \cdot \bftheta^s \right) \ dV_{\bfx_0} \\
		& + \epsilon_0 \int_{\R^3} \grad_{\bfx} \phi \cdot \grad_{\bfx} \psi \ dV_{\bfx}
\end{split}
\end{equation}
The nonlocal integral over $\R^3$ can be transformed by multiplying $\divergence_{\bfx} \grad_{\bfx} \psi = \divergence_{\bfx} \bfq$ by $\phi$ and integrating over $\R^3$.
Using integration by parts on the left side and then pulling it back to the reference gives:
\begin{equation}
	\int_{\Omega_0}\grad_{\bfx} \phi \cdot \bfq J \ dV_{\bfx_0} = \int_{\R^3} \grad_{\bfx} \phi \cdot \grad_{\bfx} \psi \ dV_{\bfx}
\end{equation}
where $J$ is the Jacobian of the deformation.
This is substituted for the nonlocal integral in (\ref{eq:grad_flow}) to get:
\begin{equation}
\label{eq:grad_flow-1}
\begin{split}
	\int_{\Omega_0} \left (\dot{\bfu}\cdot \bfv + \sum_{s} \dot{\bfzeta}^s \cdot \bftheta^s \right ) \ dV_{\bfx_0}
		=
		&
	\int_{\partial{\Omega_0}}\frac{\partial{W}}{\partial{(\grad_{\bfx_0} \bfu)}} : \bfn \bfv \ dS_{\bfx_0} 
	- \int_{\Omega_0} \left(\divergence_{\bfx_0} \frac{\partial{W}}{\partial{(\grad_{\bfx_0} \bfu)}}\right) \cdot \bfv \ dV_{\bfx_0}
		\\
		&
	+\int_{\Omega_0}\sum_{s} \frac{\partial{W}}{\partial{\bfzeta^s}}\cdot \bftheta^s \ dV_{\bfx_0}
	+ \epsilon_0 \int_{\Omega_0} \grad_{\bfx} \phi \cdot \bfq J  \ dV_{\bfx_0}
\end{split}
\end{equation}

We now examine the relation between $\bfq$ and the variations $\bfv$ and $\bftheta$.
The polarization density $\bfp$ is defined by
\begin{equation}
\label{eq:polar_def_2}
	\bfp=\frac{1}{\det(\bfF) V_0} \sum_{s} Q^s \bfF \bfzeta^s
\end{equation}
Note that $\bfzeta^s$ is the position of species $s$ in the reference configuration, while $\bfp$ is the polarization density in the current configuration.
Similarly, $V_0$ is the volume of the unit cell in reference configuration.
Hence, both $V_0$ and $\bfzeta^s$ are pushed forward to the current.

Taking the variation of $\bfp$
\begin{equation}
\label{eq:polar_def_q}
	\bfp+ \eta \bfq= \frac{1}{\det(\bfF + \eta \bfG) V_0} \sum_{s} Q^s (\bfF+\eta \bfG) (\bfzeta^s+ \eta \bftheta^s)
\end{equation}
where $\bfG:=\grad_{\bfx_0} \bfv$.
Noting that $(\det(\bfF+\eta \bfG))^{-1} = \det(\bfF)^{-1} \det(\bfI+\eta \bfF^{-1} \bfG)^{-1} = \det(\bfF)^{-1} \left( 1-\eta \tr(\bfF^{-1}\bfG)- O(\eta^2) \right)$, we find:
\begin{equation}
\label{eq:q_def}
	\bfq =  \frac{1}{\det(\bfF) V_0} \sum_{s} Q^s \left( \bfF \bftheta^s+  \bfzeta^s \cdot \grad_{\bfx_0} \bfv  - \tr(\bfF^{-1} \grad_{\bfx_0} \bfv ) \bfF \bfzeta^s \right)
\end{equation}
after ignoring terms of order higher than linear in $\eta$.
Using this expression for $\bfq$, we obtain the following expression for $\int_{\Omega_0}\grad_{\bfx}\phi \cdot \bfq J \ dV_{\bfx_0}$ after using integration-by-parts and the divergence theorem:
\begin{equation}
\label{eq:q_energy}
\begin{split}
	\frac{1}{V_0} \sum_{s} Q^s \int_{\Omega} \grad_{\bfx}\phi \cdot \left(\bfF \bftheta^s \right) \ dV_{\bfx_0}
		+
	\frac{1}{V_0} \sum_{s} Q^s \left(\int_{\partial{\Omega_0}} \grad_{\bfx}\phi \bfzeta^s : \bfn \bfv \ dS_{\bfx_0}
	 -\int_{\Omega_0}\divergence_{\bfx_0}\left(\grad_{\bfx}\phi \bfzeta^s \right) \cdot \bfv \ dV_{\bfx_0} \right)
		\\
		-
	\frac{1}{V_0} \sum_{s} Q^s \left( \int_{\partial{\Omega_0}}\grad_{\bfx}\phi \cdot \left( \bfF \bfzeta^s \right) \bfF^{-T} : \bfn \bfv \ dS_{\bfx_0}
	 - \int_{\Omega_0}\divergence_{\bfx_0}\left( \grad_{\bfx}\phi \cdot \left( \bfF \bfzeta^s \right) \bfF^{-T} \right) \cdot \bfv \ dV_{\bfx_0} \right)
\end{split}
\end{equation}

Defining $\bfA^s := \grad_{\bfx}\phi \bfzeta^s - \grad_{\bfx}\phi \cdot \bfF\bfzeta^s \bfF^{-T}$, collecting terms in (\ref{eq:q_energy}, \ref{eq:grad_flow-1})  and localizing using the arbitrariness of the variations gives us the compact form:
\begin{equation}
\label{eq:u_gov}
	\dot{\bfu} = \divergence_{\bfx_0} \left(\frac{\partial{W}}{\partial{\grad_{\bfx_0} \bfu}}+\epsilon_0\sum_{s} \frac{Q^s}{V_0}\bfA^s\right)
	\quad , \quad
	\dot{\bfzeta^s} = -\frac{\partial{W}}{\partial{\bfzeta^s}} -\epsilon_0\frac{Q^s}{V_0} \grad_{\bfx} \phi \cdot \bfF
\end{equation}
with the boundary term:
\begin{equation}
\label{eq:boundary_gov}
	\left(\frac{\partial{W}}{\partial{\grad_{\bfx_0} \bfu}} +  \epsilon_0\sum_{s}\frac{Q^s}{V_0}\bfA^s \right) \cdot \bfn =0
\end{equation}
These equations provide the gradient descent equations for the fields $\bfu$ and $\bfzeta^s$.
The equation for $\bfu$ involves a standard mechanical stress as well as a so-called Maxwell or electromechanical stress.
The equation for $\bfzeta^s$ is a local equation (i.e., not a PDE), but is coupled through $W$ and the electric field to the nonlocal electrostatics and the PDE for momentum balance.

%%%%%%%%%%%%%
%%%%%%%%%%%%%
%%%%%%%%%%%%%
%%%%%%%%%%%%%

\subsubsection{Electromechanical Transformations to the Reference Configuration}

For simplicity, we solve the electrostatic Poisson equation in the current configuration for the electric potential $\phi(\bfx)$ and compute the electric field $\grad_{\bfx}\phi(\bfx)$ and rewrite $\bfx=\bfx(\bfx_0)$ from the deformation map.
In our setting, the energy (both short- and long-range) and electric fields are based entirely on the atomic position in the current configuration.
Therefore, in common with much of continuum mechanics, the reference configuration can be considered simply a change of variables for bookkeeping convenience.
It follows that the definition of electrical quantities in the reference configuration are not essential.
This can also be readily observed from (\ref{eq:q_energy}); the physical quantity of interest there is $\grad_{\bfx}\phi(\bfx)$, and while it is certainly possible to define the electric field in the reference, it is not physically important.

This has led to a number of different proposals for referential electrical quantities in the literature.
For instance, our (\ref{eq:polar_def_2}) suggests that the reference polarization $\bfp_0(\bfx_0)$ is given by $\bfp(\bfx(\bfx_0)) = \det(\bfF)^{-1} \bfF \bfp_0(\bfx_0)$, i.e. the polarization transforms as material line elements that carry charges but with an additional factor accounting for volume changes.

In \cite{xiao-bhatta}, they assume instead $\bfp(\bfx(\bfx_0)) = \det(\bfF)^{-1} \bfp_0(\bfx_0)$, and further assume that the reference electrostatic potential $\phi_0(\bfx_0) = \phi(\bfx(\bfx_0))$; this gives them that the electric field transforms as line elements using the identity $\grad_{\bfx} = \bfF^{-T}\grad_{\bfx_0}$.
This is consistent with the differential geometric notion of the electric field as a $1$-form, i.e. it is a quantity that is integrated along lines.

In \cite{ponte-siboni}, following \cite{dorfmann-ogden}, they use yet another transformation; they work with electric displacement $\bfD$ and electric field $\bfE$ as primary variables rather than polarization.
These variables are motivated by the fact that the continuum problem posed in polarization and electric potential leads to a saddle-point variational problem,  whereas the minimization structure is preserved in $\bfD$ and $\bfE$.
The electrostatic equations in these variables are $\divergence \bfD = 0$ and $\curl \bfE = 0$.
For $\bfE$, these imply it should transform as a material line element and matches with \cite{xiao-bhatta} as well as the differential geometric notion of the electric field as a $1$-form.
For $\bfD$ however, these imply a transformation $\bfD = \det(\bfF)^{-1} \bfF^{-T} \bfD_0$ following the differential geometric notion of the electric displacement as a $2$-form, i.e. a quantity that is integrated over surfaces; this is obviously identical to the standard stress transformation.
This differential geometric notion is also implicitly exploited in Section 4 of \cite{kohn-shipman} in finding the appropriate averaging for the different field quantities.
An essential practical advantage of this transformation that is exploited by \cite{ponte-siboni} in their homogenization analysis is that the equations retain their structure in the reference configuration, i.e. $\divergence_{\bfx_0} \bfD_0 = 0$ and $\curl_{\bfx_0} \bfE_0 = 0$.

While the transformations proposed by these other workers are physically appealing for the reasons mentioned above, the transformation implied by the microscopic model of the polarization is also physically motivated.
These different transformations are not consistent with preserving the relation $\bfD = \epsilon_0 \bfE + \bfp$ between corresponding quantities in the reference.

%%%%%%%%%%%%%%%%%%%%%
%%%%%%%%%%%%%%%%%%%%%
%%%%%%%%%%%%%%%%%%%%%
%%%%%%%%%%%%%%%%%%%%%

\subsection{Local Quasicontinuum for Multi-lattices with Short-Range Interactions}

We use standard finite element linear shape functions, $N_a$ defined at the nodes $a$ to approximate the displacements $\bfu$ of the lattice vectors:
\begin{equation}
\label{eq:u_shape}
	\bfu \approx \sum_{a} \bfu_a N_a \Rightarrow \grad_{\bfx_0} \bfu \approx \sum_{a} \bfu_a \grad_{\bfx_0} N_a
\end{equation}
where $\bfu_a$ are the nodal displacements.
Defining $\bfB$ as the Piola-Maxwell stress tensor,
\begin{equation}
\label{eq:div_stress}
	\divergence_{\bfx_0} \bfB = \divergence_{\bfx_0} \left(\frac{\partial{W}}{\partial{\grad_{\bfx_0} \bfu}}+\epsilon_0\sum_{s}\frac{Q^s}{V_0}\bfA^s\right)
\end{equation}
we can write (\ref{eq:u_gov}) as  ${\bf 0} = \divergence_{\bfx_0} \bfB$.
Standard nonlinear finite element methods can now be used; the key difference is that we need to compute the electromechanical contribution to the stress at every iteration.
At the same time, we also iterate with respect to $\bfzeta^s$ noting that our interpolation for these variables is piecewise constant, i.e. constant in a given element.

There are two short-range calculations: first, the Piola-Maxwell stress tensor, and second, the minimization over $\bfzeta^s$.
Following the complex local QC method \cite{Tadmor}, the deformation gradient and $\bfzeta^s$ are related to the atomic displacements through the Cauchy-Born rule.  We assume that the energy of each element can be approximated by assuming a homogeneous deformation of the crystal through the deformation gradient.
Given an interatomic potential $U$, we use
\begin{equation}
\label{eq:pair_potential}
	\frac{\partial{W}}{\partial{\grad_{\bfx_0} \bfu}} = \frac{1}{2V}\sum_{\text{atoms in cutoff}} \frac{\partial{U}}{\partial{\bfr}}\frac{\partial{\bfr}}{\partial{\grad_{\bfx_0} \bfu}}
\end{equation}
where $\bfr$ are the positions of the atoms and are obtained from $\bfu$ and $\bfzeta^s$.
Similarly, we can compute
\begin{equation}
\label{eq:pair_potential2}
	\frac{\partial{W}}{\partial{\boldsymbol{\zeta}^s}} = \frac{1}{2V}\sum_{\text{atoms in cutoff}} \frac{\partial{U}}{\partial{\bfr}}\frac{\partial{\bfr}}{\partial{\boldsymbol{\zeta}^s}}
\end{equation}
The minimization over the nodal values of $\bfu$ and the element values of $\bfzeta^s$ are conducted in a coupled manner.

%%%%%%%%%%%%%%%%%%%%%
%%%%%%%%%%%%%%%%%%%%%
%%%%%%%%%%%%%%%%%%%%%
%%%%%%%%%%%%%%%%%%%%%

\subsection{Coarse-graining of the Long-Range Electrostatic Interactions}

The energy minimization calculation in (\ref{eq:u_gov}) requires the electric field for a given distribution of charge, or equivalently for a given $\bfp$.
As noted above, we have a constant $\bfF$ and $\bfzeta^s$ field in each element, i.e. they are discontinuous only along element boundaries.
Consequently, $\bfp$ as defined in (\ref{eq:polar_def_2}) is also constant in each element and discontinuous along the element boundaries.
These discontinuities in polarization result in surface charge densities $(\bfp_1 - \bfp_2) \cdot \bfn$.
We compute the fields due to these relatively simple charge distributions using direct Greens function integrations.

%%%%%%%%%%%%%%%%%%%%%%
%%%%%%%%%%%%%%%%%%%%%%
%%%%%%%%%%%%%%%%%%%%%%
%%%%%%%%%%%%%%%%%%%%%%

\section{Crystal Free Surface Subject to Inhomogeneous Electric Fields}
\label{sec:examples}

We apply the methodology described above to a simple setting of a crystal with a free surface subject to an inhomogeneous external electric field due to a point charge above the surface.
We use short-range potentials based on the bi-species Lennard-Jones model used for Ni-Mn by \cite{Hildebrand}, giving a tetragonal neutral lattice with a body-centered ion and a charged shell.
The polarization is oriented along the tetragonal direction when external fields are absent, thereby providing a spontaneous polarization.
This provides a simple model of many widely-used perovskite ferroelectrics such as barium titanate and lead titanate.

Given that this is a model material rather than numerically accurate, we aim to elucidate the physics of electromechanics.
There are two independent ratios of interest.
One is the strength of the ionic charges v. the strength of the external charge.
The other is the strength of ionic interactions v. the strength of the bonded interactions.
We explore the effect of the former by increasing external charge while holding ionic charges fixed.
We explore the effect of the latter by scaling the electrostatic interactions.
This enables us to test a class of materials that range from non-ionic to ionic.

We consider a specimen with the spontaneous polarization oriented tangential to the surface.
A single point charge is placed near the surface.
In all of these examples, atomic resolution is used in regions of interest -- particularly beneath the external charges -- while coarser resolution is provided everywhere else.
A fine mesh is introduced near the point charge with coarsening throughout the rest of the body.
Fig. \ref{fig:mesh} shows the full mesh, while Fig. \ref{fig:mesh_zoom} shows the zoomed in atomistically resolved portion of the mesh with all atoms plotted in the current configuration.
The black atoms are nodes while the green atoms are constrained by the interpolation.

We conduct three calculations with the electrostatic interactions scaled by factors of $1, 2, 4$ respectively, while holding the charge location and strength fixed.
For the given charge strength, a scaling by $1$ produces a surface distortion near the point charge, while the scaling by $2$ and $4$ produce nucleation-like events\footnote{
A nucleation-like event refers to localized switching of polarization.  We do not expect our method to be qualitatively accurate beyond this regime.
}.
Fig. \ref{fig:scaling-1} shows the stress and polarization fields near the point charge for electrostatic scaling of $1$.
Similarly, Fig. \ref{fig:scaling-4} shows the stress and polarization for the electrostatic scaling of $4$.

\begin{figure}[ht!]
	\centering
	\includegraphics[width=\textwidth]{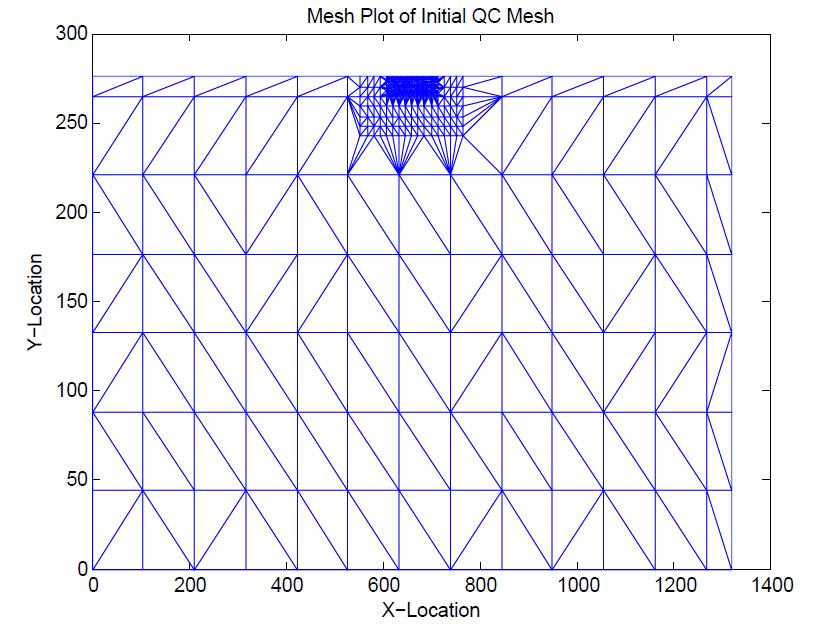}
	\caption{\small The mesh of the entire specimen.  The lengthscales are angstroms.}
	\label{fig:mesh}
\end{figure}

\begin{figure}[ht!]
	\centering
	\includegraphics[width=\textwidth]{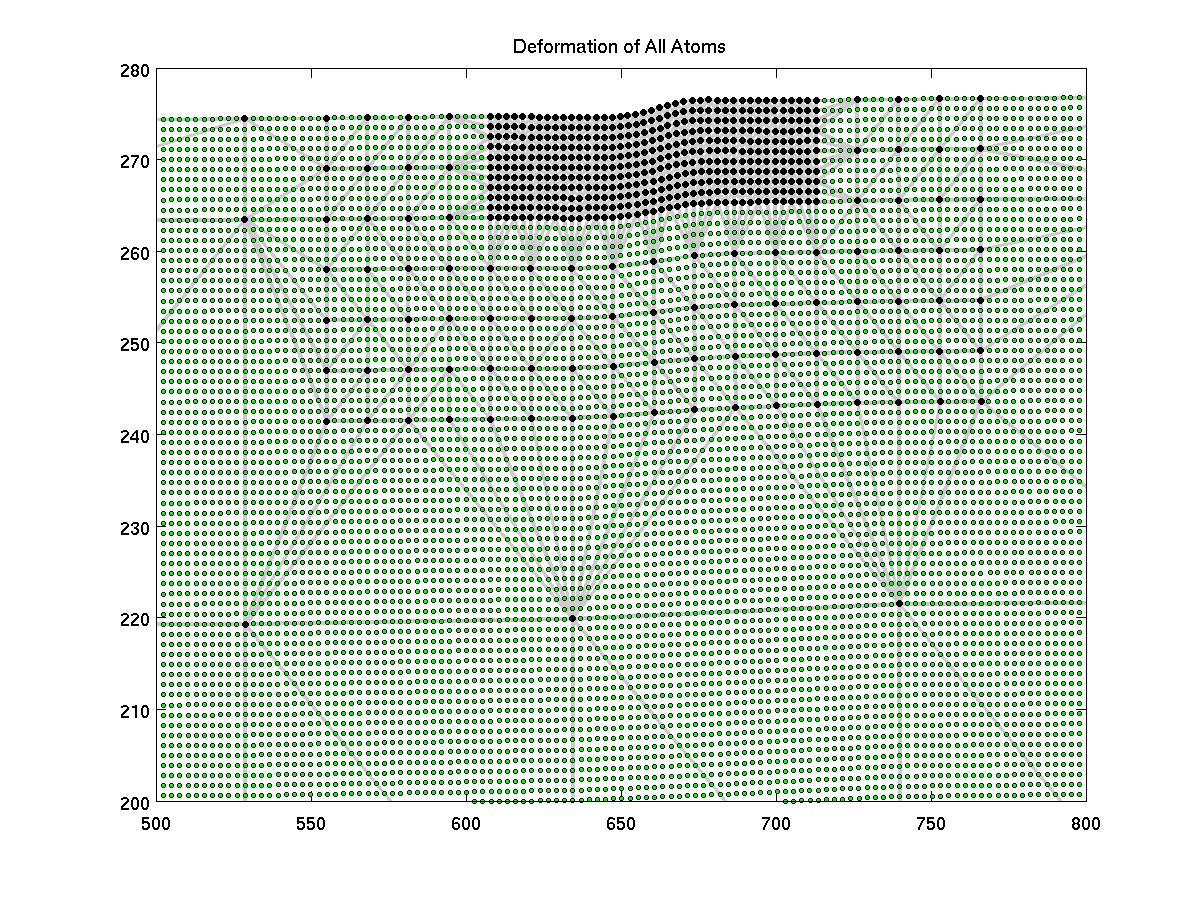}
	\caption{\small Close up view of the atomically-resolved portion of the sample.  Without applied fields and loads, the surface of the specimen is flat.  The surface feature is caused by an applied electric field.}
	\label{fig:mesh_zoom}
\end{figure}

The next set of calculations replaces the single point charge by multiple point charges ($2$ and $4$ respectively), keeping the total charge the same as the case of the single point charge and with spacing between charges on the order of the height above the surface.
In subsequent sets of simulations, these charges were moved closer to the surface until eventually a nucleation-like event occurred.
Stress plots are provided for the $2$ and $4$ charges at various distances from the free surface in Figs. \ref{fig:stress-2chargefar}-\ref{fig:stress-4chargeclose}.

An interesting feature is the presence of two stress lobes beneath the surface of the material.
To further investigate this, we replace the point charge by a dipole (two charges with opposite signs).
The resulting deformation as seen in Figure \ref{fig:stress-dipole} is a depression instead of the up-down pattern seen in loadings with charges of the same sign.  Additionally, the two stress lobes previously seen disappear and merge into a single stress lobe beneath the surface.

The stress lobes in non-dipole loadings do not appear to match continuum phase-field calculations, Fig. \ref{fig:Lun} following \cite{Lun}.
These continuum calculations are based on linearized electromechanics whereas the atomistic calculations in this paper are inherently large-deformation.
The small-deformation approximation implies that electric fields are computed in the reference as well as potential large rotations being ignored.
Further, the material model is also linear elastic.
To examine if these assumptions are responsible for the inability of continuum models to capture this feature, we perform the following tests.
First, we use precisely our approach except that the electric fields are computed in the reference as in linearized electromechanics.
The resulting stress field still shows the double-lobe structure, though with a slightly reduced magnitude (Fig. \ref{fig:finite_elect}).
Second, we examine the geometric linearization of the strains; Fig. \ref{fig:normal_strain22} shows a plot of the normalized difference between the non-linear and linearized $\epsilon_{22}$ strain measures which are far too small to be the cause.
Third, we test against a fully atomic level description without any coarse-graining.
While the system size is small due to computational limitations of the electrostatics, we see a 2-lobe structure Fig \ref{fig:full-atomic}.
Algorithmic reasons prevent us from computing the stress and deformation gradient in the fully-atomic setting, so we instead examine the change in energy of each atom from the perfect crystal.
While qualitative, this calculation supports the view that the double-lobe structure is an atomic effect that we are able to capture despite the local QC approximation.

Regarding the lack of agreement with continuum phase-field, we conjecture that this may simply be due to the gradient penalty preventing the development of fine-scale features such as the double-lobe.
Our approach can be viewed as a phase-field method with no gradient penalty.
In typical phase-field modeling, the gradient penalty is taken to significantly larger than justified by domain wall widths, because the goal there is to predict microstructure and not defect structure.
This may cause fine-scale features such as the double-lobe to be washed out.
Of course, it is also debatable whether our fully local QC model is appropriate to model such fine-scale features.

Finally, we examine the effect of a mechanical indenter that is pressed into the surface.
Plots of the stress and polarization are in Fig. \ref{fig:stress-indent}.
These are largely as expected.

We note that in these calculations, the electrostatics  appears to induce a lengthscale.
However, neither classical electrostatics not the local QC -- which is essentially classical continuum mechanics with an atomically-informed constitutive model -- has an intrinsic lengthscale.
The induced lengthscale from the electrostatics is nonlocal though problem-dependent, i.e. it depends on sample size, boundary conditions, and so on.

%%%%%%%%%%%%%%%%%%%%%%
%%%%%%%%%%%%%%%%%%%%%%
%%%%%%%%%%%%%%%%%%%%%%
%%%%%%%%%%%%%%%%%%%%%%

\section{Discussion}
\label{sec:conclusion}

We have presented a multiscale atomistic method for ionic solids and other materials where electrostatic interactions are long-range.
Our approach is based on using the polarization field as a multiscale mediator between atomic-level rapidly-varying charge distributions and the continuum scale electrical quantities.
Our coarse-graining strategy relies on ideas developed by \cite{james-muller} and others that followed them \cite{schlomerkemper-schmidt,xiao-bhatta,puri-bhatta}.
The method that we have presented enables QC and other multiscale methods to go beyond purely short-range interactions, that are characteristic of many structural materials, to functional and electronic materials of interest today.

The method presented here is a first attempt towards the goal of understanding the multiscale electromechanics of defects in ionic solids.
Consequently, there remain many important open questions.
\begin{enumerate}

\item The presentation here is based on energy minimization which implies zero temperature; recent methods based on extending this to finite temperature in the setting of short-range interactions can perhaps be adapted to our setting \cite{kulkarni}.
An alternate approach is \cite{tadmor-EffHamil}, where the powerful method of Effective Hamiltonians has been applied to study the finite temperature behavior of ferroelectrics in the local QC setting.

\item We have started from an atomic viewpoint in this paper whereas QC methods have been demonstrated for the orbital-free density functional theory -- a version of (ground state) density functional theory that is restricted to metals with mobile electrons.
However, variants of density functional theory have difficulty with charged defects as rigorously demonstrated in \cite{cances-ehrlacher}.
In addition, density functional theory has important qualitative failings in computing bandgaps, van der Waals interactions, and so forth \cite{DFT-vdW-crap}.
Therefore, it appears to be more useful to begin from the atomic level with well-calibrated and trustworthy potentials.

\item We have worked within the local QC setting which can alternately be considered a standard finite element approach with the constitutive relation being drawn from atomistics rather than a prescribed function.
An important further step is to couple the coarse-grained model with a region with truly atomic resolution in the vicinity of the defect.
In this regard, the key coarse-graining ideas presented here carry through to that setting.
This open question is an area of our current research and we expect to report on this in the near future.

\item Our use of linear interpolation restricts us from capturing potentially important effects based on strain gradients, in particular the phenomenon of flexoelectricity that can be relevant at small scales \cite{pradeep-flexo}.
However, this phenomenon is relevant at large strain gradients that typically can only occur in localized regions of the sample.
Therefore, a QC method that fully resolves the defect region can potentially capture this phenomenon even with linear interpolations.

\item Our coarse-graining of the electrostatics in this paper assumed complete separation of scales.
However, for a real calculation, this will naturally not be true.
It is therefore important to study and quantify the errors when the separation of scales is large but not complete.
This will enable the construction of an algorithm with controlled error tolerance.
Directly related to this is the derivation of a rigorous limit as opposed to the formal presentation in this paper.
These related open questions are an area of our current research, following techniques in \cite{schlomerkemper-schmidt}.

\end{enumerate}

The example that we have studied of the double-lobe structure beneath a point charge provides an important motivation for the further development of our method.
Atomic-scale features such as the double-lobe are of importance to phenomena such as domain nucleation and cannot be captured by continuum models.
On the other hand, the limited size of our admittedly-unoptimized brute force atomic calculation shows the need for coarse-graining efforts.

Finally, we have examined the issue of whether polarization is unique.
As noted above, evaluating the classical definition of polarization provides a unit cell-dependent quantity.
This is obviously a potential disaster for a coarse-graining scheme based on the polarization field.
The materials physics orthodoxy appeals to quantum mechanical concepts to fix a unique value of the polarization \cite{Resta-Vanderbilt}.
Continuum mechanicians have used variational notions to achieve similar ends \cite{puri-bhatta}.
The view that we have taken in this paper is that there is simply no need to have a unique value for the polarization.
When partial unit cells that provide boundary charges are accounted for consistently, the coarse-grained electric fields and other relevant quantities are independent of the choice of unit cell; the polarization is merely a multiscale intermediary.
In continuum theories of electromechanical solids, e.g. \cite{xiao-bhatta,shu-bhatta}, the polarization appears in the standard local free energy density, not only in the electrostatics.
In the local free energy density, the polarization can be considered as a quantity that tracks the atomic positions rather than relevant to electrostatics.
In that perspective, the specific choice of unit cell is irrelevant.
Broadly, the view advocated in this paper is in the spirit of classical continuum mechanics: the polarization field simply provides some information about the atomic-level, and one can take any choice as long as consistent transformations between energy, kinematics, and boundaries are respected.
The immediate analog in continuum mechanics is the freedom in the choice of reference configuration and the corresponding value of the deformation field and strain energy density response function as long as care is taken to define suitable transformations between different choices.

%%%%%%%%%%%%%%%%%%%%%%
%%%%%%%%%%%%%%%%%%%%%%
%%%%%%%%%%%%%%%%%%%%%%
%%%%%%%%%%%%%%%%%%%%%%

\section*{Acknowledgments}
\label{sec:acknowledgments}

We thank ARO Numerical Analysis for financial support through a Young Investigator grant (W911NF-12-1-0156).
Jason Marshall also acknowledges support from the Northrop Graduate Fellowship Award from Carnegie Mellon University.
Kaushik Dayal also acknowledges support from AFOSR Computational Mathematics (FA9550-09-1-0393) and AFOSR Young Investigator Program (FA9550-12-1-0350).
Kaushik Dayal thanks the Hausdorff Research Institute for Mathematics at the University of Bonn for hospitality.
This research was also supported in part by the National Science Foundation through TeraGrid resources provided by Pittsburgh Supercomputing Center.
We thank Richard D. James, Saurabh Puri, and Yu Xiao for useful discussions.

%%%%%%%%%%%%%%%%%%%%%%
%%%%%%%%%%%%%%%%%%%%%%
%%%%%%%%%%%%%%%%%%%%%%
%%%%%%%%%%%%%%%%%%%%%%

\appendix

\section{Transformation of Unit Cell in a Crystal Lattice to Obtain a Unit Cell with One Face Parallel to a Given Plane}
\label{sec:appendix}

We show in this section that any crystal lattice can be described by a unit cell with one face parallel to  a given plane.
This is an essential ingredient of our proof that the change in polarization due a change in the lattice unit cell is balanced by a corresponding change in surface charge density.

We first show the following proposition.

\begin{proposition}
Consider a crystal described by lattice vectors $\{ \bff_1, \bff_2, \bff_3 \}$ and reciprocal lattice vectors $\{ \bff^1, \bff^2, \bff^3 \}$.
Consider a rational plane\footnote{
A rational plane has a normal $\bfn$ that can be represented $\bfn = \sum_i n_i \bff^i$ with $n_i$ integers.
}
 with normal $\bfn$.

This lattice can also be described by lattice vectors $\{ \bfg_1, \bfg_2, \bfg_3 \}$, where $\bfg_2 \perp \bfn$, $\bfg_3 = \bff_3$, and $\bfg_1$ will be shown below to exist.
\end{proposition}

\begin{proof}
Result 3.1 of \cite{bhatta-book} states that a crystal described by lattice vectors $\{ \bff_1, \bff_2, \bff_3 \}$ can equivalently be described by lattice vectors $\bfg_1, \bfg_2, \bfg_3$ if and only if these satisfy $\bfg_i = \sum_j \mu_i^j \bff_j$ with $\mu_i^j$ being a $3 \times 3$ matrix of integers with determinant $\pm 1$.

We have $\bfg_3=\bff_3$, therefore $\mu_3^1 = \mu_3^2 = 0$ and $\mu_3^3=1$.
Set $\bfg_2 = M_1 \bff_1 + M_2 \bff_2$ with $M_1 = -\frac{n_2}{\gcd(n_2,n_1)}$ and $M_2 = \frac{n1}{\gcd(n_2,n_1)}$.
This choice ensures the following: $\bfg_2 \perp \bfn$, $M_1$ and $M_2$ are integers, and $\gcd(M_1, M_2)=1$.
It implies that $\mu_2^1 = M_1, \mu_2^2 = M_2, \mu_2^3=0$.

The remaining step to complete the proof is to show that we can find $\bfg_1$ which satisfies $\bfg_i = \sum_j \mu_i^j \bff_j$ where $\mu_i^j$ has the restrictions mentioned above.
Write $\bfg_1 = N_1 \bff_1 + N_2 \bff_2 + N_3 \bff_3$.
Then $\det \mu = 1$ gives us the condition $N_1 M_2 - M_1 N_2 = 1$.
The extended Euclidean algorithm provides the existence of integer solutions $N_1$ and $N_2$ to this equation \cite{euclidean-algorithm}.
In general, this algorithm provides integer solutions $m,n$ to the equation $a m + b n = \gcd(a,b)$ where $a, b$ are given integers.
While the algorithm is constructive and ensures existence of solutions, it does not provide explicit forms for the solutions.

Therefore, there exists $\bfg_1 = N_1 \bff_1 + N_2 \bff_2 + N_3 \bff_3$ where $N_1, N_2$ are obtained from the algorithm above and $N_3$ can be arbitrary.
\end{proof}

We wish to show that given a crystal described by lattice vectors $\{ \bff_1, \bff_2, \bff_3 \}$, we can equivalently describe the crystal by lattice vectors $\{ \bfh_1, \bfh_2, \bfh_3 \}$ with $\bfh_2 \perp \bfn$ and $\bfh_3 \perp \bfn$.
The vector $\bfn$ is the unit normal to a rational plane.

We simply apply the proposition above twice in succession.
First, we use the proposition to transform from  $\{ \bff_1, \bff_2, \bff_3 \}$ to  $\{ \bfg_1, \bfg_2 \perp \bfn, \bfg_3 = \bff_3 \}$.
We then use the proposition again to transform from $\{ \bfg_1, \bfg_2, \bfg_3 \}$ to $\{ \bfh_1, \bfh_2 = \bfg_2, \bfh_3 \perp \bfn\}$.
Therefore, the final description has $\{ \bfh_1, \bfh_2 \perp \bfn, \bfh_3 \perp \bfn \}$.
The condition $\det \mu = \pm 1$ ensures linear independence; in fact, it preserves the volume of the unit cell.

%%%%%%%%%%%%%%%%%%%%%%%%%%%%%%%%%%%%%%%%%
%%%%%%%%%%%%%%%%%%%%%%%%%%%%%%%%%%%%%%%%%
%%%%%%%%%%%%%%%%%%%%%%%%%%%%%%%%%%%%%%%%%
%%%%%%%%%%%%%%%%%%%%%%%%%%%%%%%%%%%%%%%%%

\begin{figure}[h!]
	\centering
	\includegraphics[width=160mm]{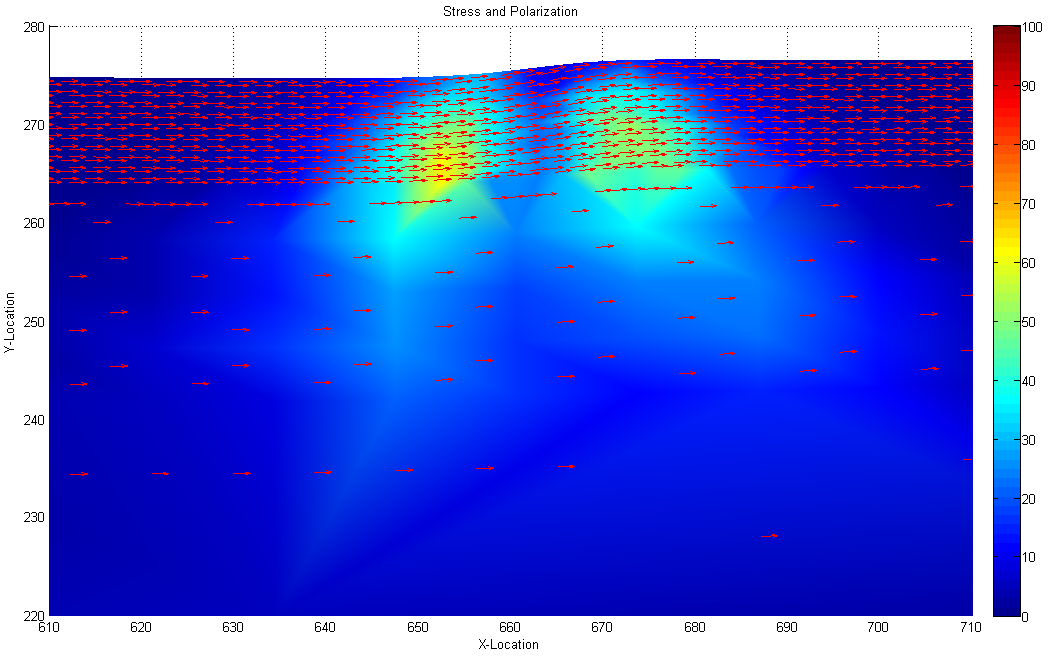}
	\caption{\small Stress and polarization due to single point charge with an electrostatic scaling of $1$.}
	\label{fig:scaling-1}
\end{figure}
% 
% \begin{figure}[ht!]
% 	\centering
% 	\includegraphics[width=130mm]{Pictures/polarization_x1.jpg}
% 	\caption{\small $1-$component of polarization due to single point charge with an electrostatic scaling of $1$.}
% 	\label{fig:polar-x1}
% \end{figure}
% 
% \begin{figure}[ht!]
% 	\centering
% 	\includegraphics[width=130mm]{Pictures/polarization_y1.jpg}
% 	\caption{\small $2-$component of polarization due to single point charge with an electrostatic scaling of $1$.}
% 	\label{fig:polar-y1}
% \end{figure}

\begin{figure}[h!]
	\centering
	\includegraphics[width=160mm]{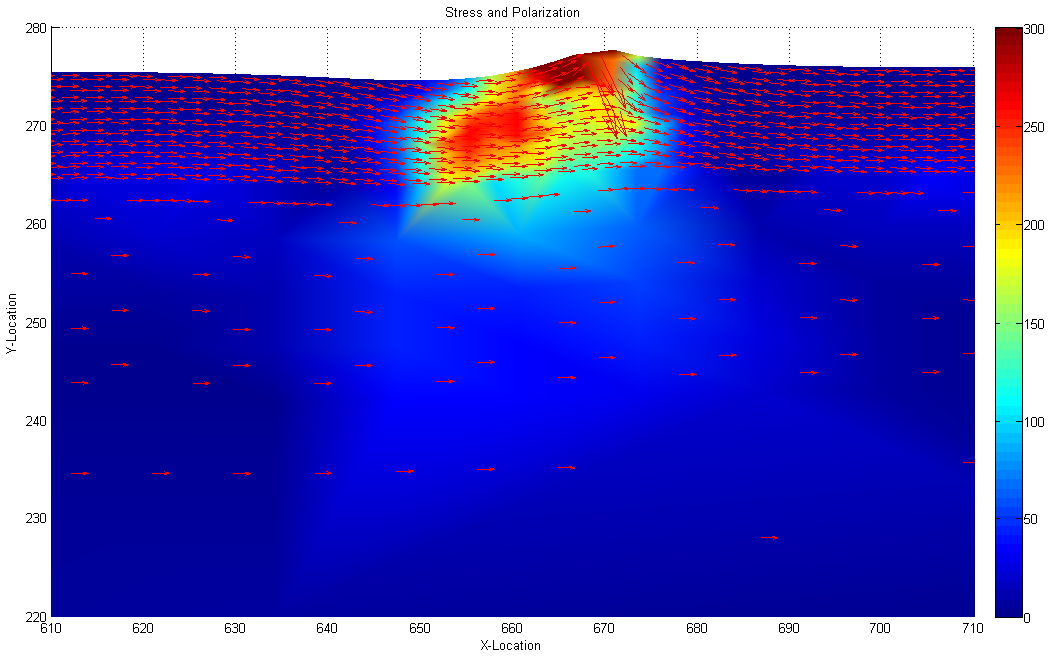}
	\caption{\small Stress and polarization due to single point charge with an electrostatic scaling of $4$.}
	\label{fig:scaling-4}
\end{figure}

% \begin{figure}[ht!]
% 	\centering
% 	\includegraphics[width=125mm]{Pictures/polarization_x4.jpg}
% 	\caption{\small $1-$component of polarization due to single point charge with an electrostatic scaling of $4$.}
% 	\label{fig:polar-x4}
% \end{figure}
% 
% \begin{figure}[ht!]
% 	\centering
% 	\includegraphics[width=130mm]{Pictures/polarization_y4.jpg}
% 	\caption{\small $2-$component of polarization due to single point charge with an electrostatic scaling of $4$.}
% 	\label{fig:polar-y4}
% \end{figure}

\begin{figure}[h!]
	\centering
	\includegraphics[width=160mm]{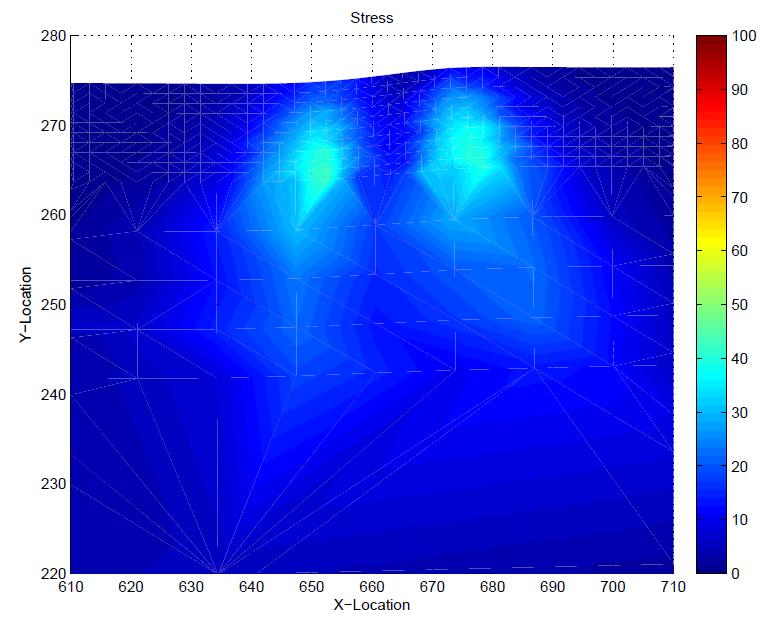}
	\caption{\small Stress due to two point charges with an electrostatic scaling of $1$.}
	\label{fig:stress-2chargefar}
\end{figure}

\begin{figure}[h!]
	\centering
	\includegraphics[width=130mm]{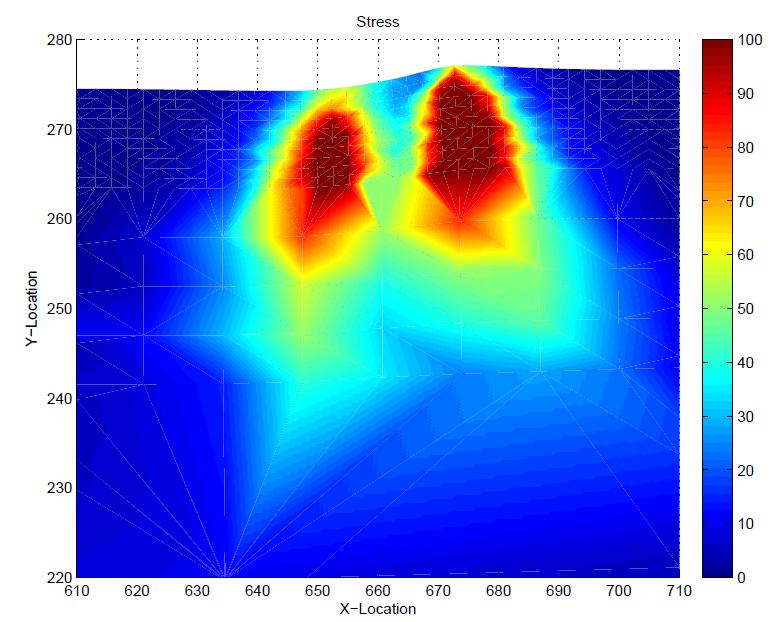}
	\caption{\small Stress due to two closely-spaced point charges with an electrostatic scaling of $1$.}
	\label{fig:stress-2chargeclose}
\end{figure}

\begin{figure}[h!]
	\centering
	\includegraphics[width=130mm]{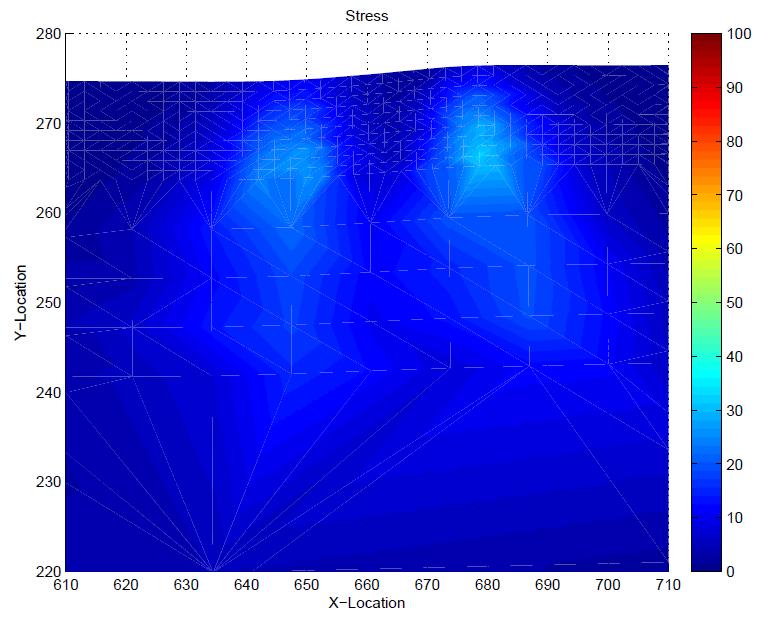}
	\caption{\small Stress due to four point charges with an electrostatic scaling of $1$.}
	\label{fig:stress-4chargefar}
\end{figure}

\begin{figure}[h!]
	\centering
	\includegraphics[width=130mm]{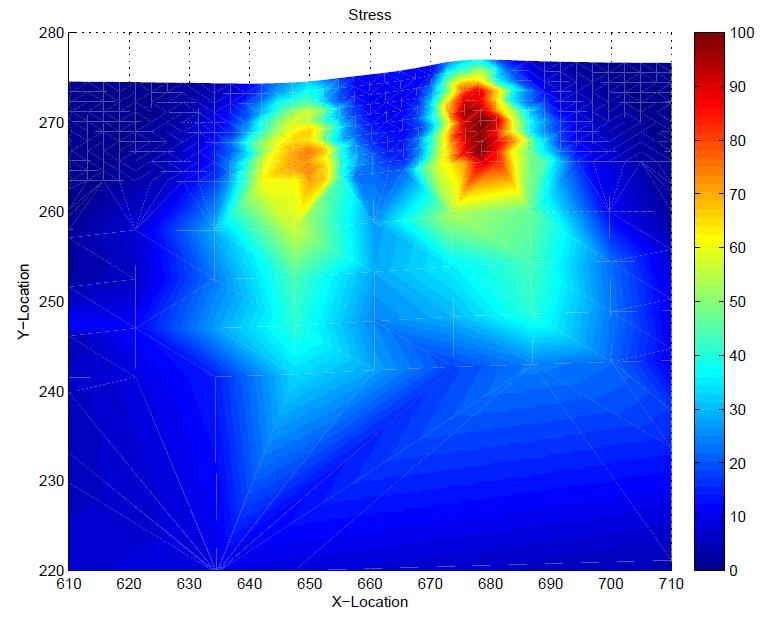}
	\caption{\small Stress due to four closely-spaced point charges with an electrostatic scaling of $1$.}
	\label{fig:stress-4chargeclose}
\end{figure}

\begin{figure}[h!]
	\centering
	\includegraphics[width=130mm]{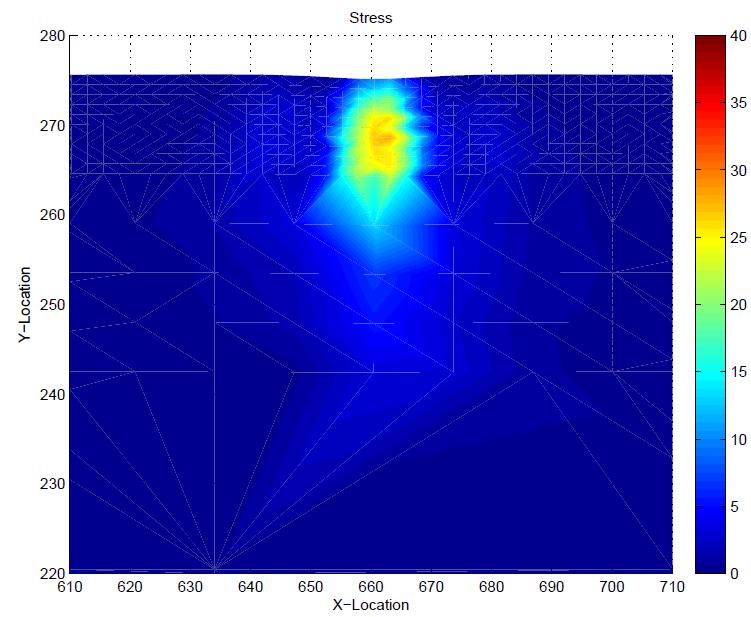}
	\caption{\small Stress due to a dipole (two point charges of opposite sign) with an electrostatic scaling of $1$.}
	\label{fig:stress-dipole}
\end{figure}

\begin{figure}[h!]
	\centering
	\includegraphics[width=130mm]{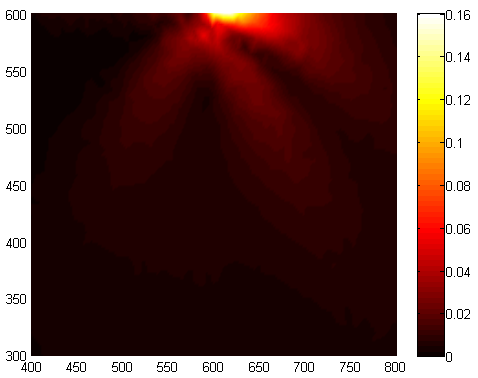}
	\caption{\small Stress due to a point charge computed using a continuum phase-field model \cite{Lun}.}
	\label{fig:Lun}
\end{figure}

\begin{figure}[h!]
	\centering
	\includegraphics[width=130mm]{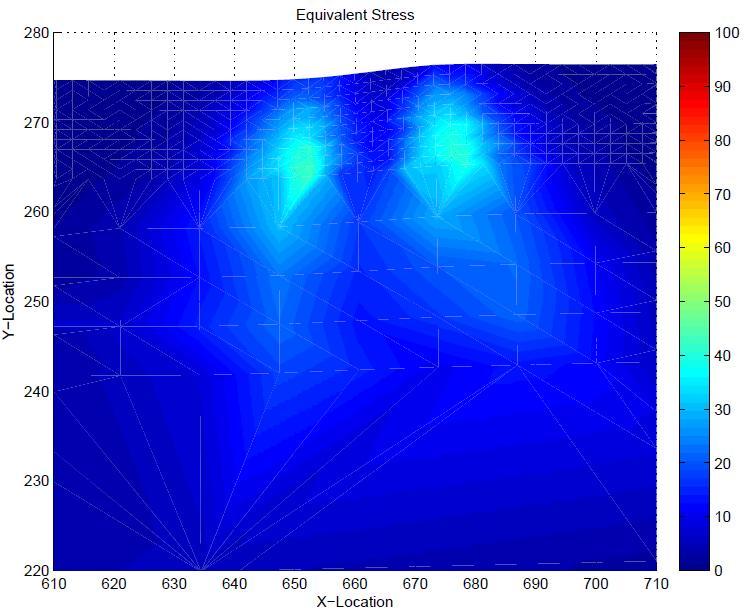}
	\caption{\small Stress due to a point charge with electrostatic fields computed in the reference configuration.}
	\label{fig:finite_elect}
\end{figure}

\begin{figure}[h!]
	\centering
	\includegraphics[width=130mm]{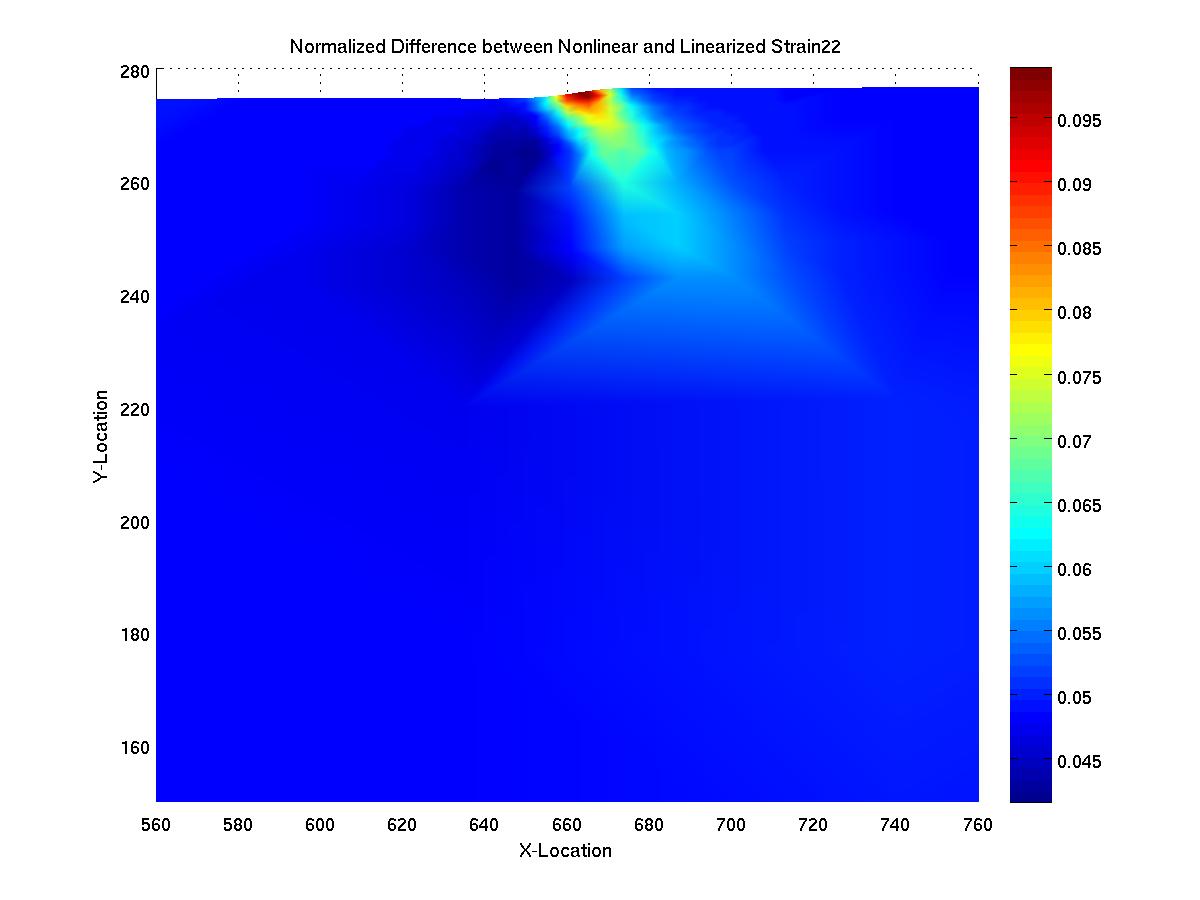}
	\caption{\small Normalized Difference between Nonlinear and Linearized Strain Measure ($\epsilon_{22}$ component).}
	\label{fig:normal_strain22}
\end{figure}

\begin{figure}[h!]
	\centering
	\includegraphics[width=130mm]{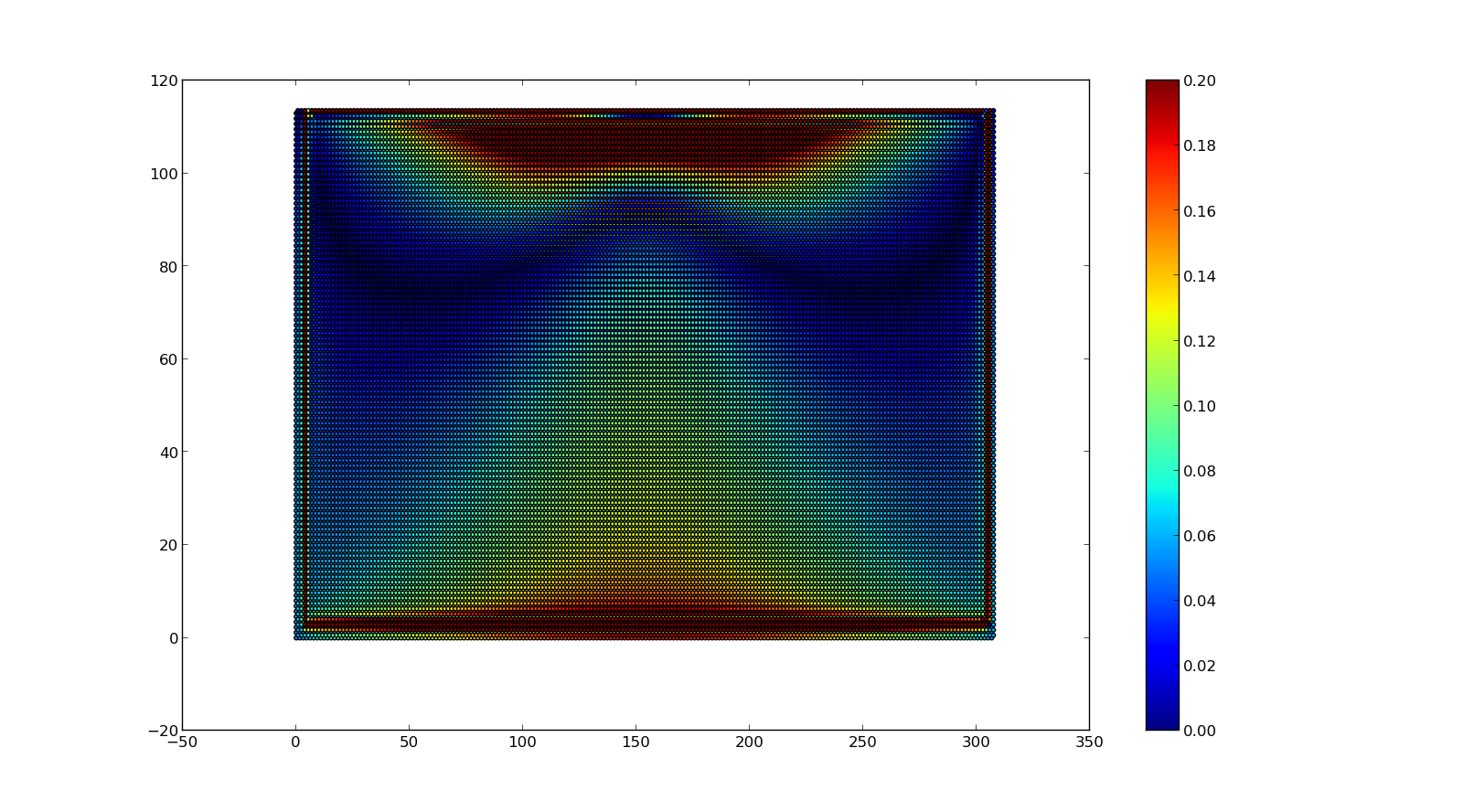}
	\caption{\small Fully atomic calculation for a sample with a point charge.  The plot shows the energy difference from the state without any external charge with the electrostatic field energy subtracted.  The entire computational domain is shown.}
	\label{fig:full-atomic}
\end{figure}

\begin{figure}[h!]
	\centering
	\includegraphics[width=130mm]{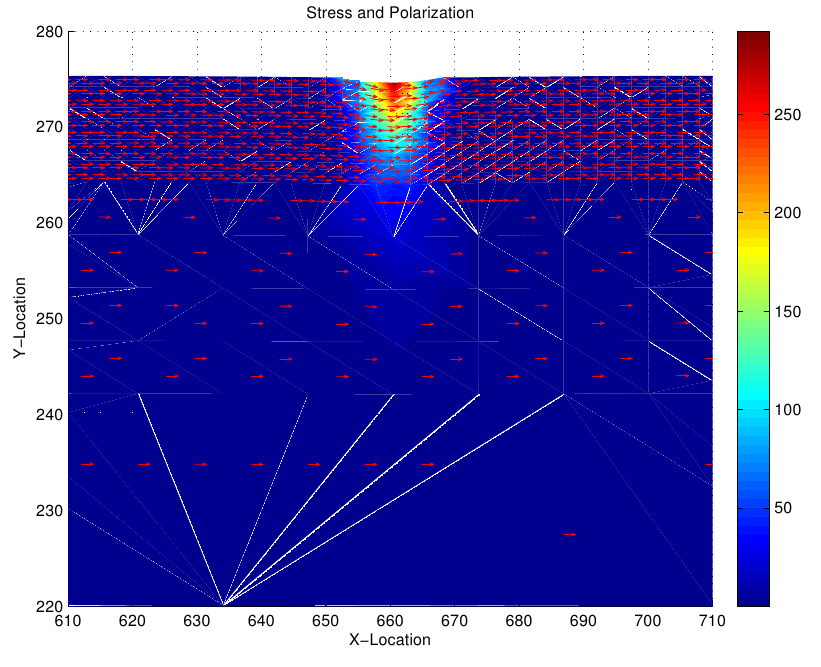}
	\caption{\small Stress and polarization due to mechanical indentation with no applied electric field.}
	\label{fig:stress-indent}
\end{figure}

%%%%%%%%%%%%%%%%%%%%%%
%%%%%%%%%%%%%%%%%%%%%%
%%%%%%%%%%%%%%%%%%%%%%
%%%%%%%%%%%%%%%%%%%%%%

\addcontentsline{toc}{section}{References}
\bibliography{mybib}
\bibliographystyle{amsalpha}

\end{document}